\documentclass[11pt]{article}

\def\showauthornotes{0}

\def\showkeys{0}
\def\showdraftbox{0}
\def\showcolorlinks{1}
\def\usemicrotype{1}
\def\showfixme{0}

\def\full{1}

\usepackage{etex}

\usepackage[l2tabu, orthodox]{nag}

\usepackage{xspace,enumerate}

\usepackage[dvipsnames]{xcolor}

\usepackage[american]{babel}

\usepackage{mathtools}

\usepackage{amsthm}

\newtheorem{theorem}{Theorem}[section]
\newtheorem*{theorem*}{Theorem}

\newtheorem{proposition}[theorem]{Proposition}
\newtheorem*{proposition*}{Proposition}
\newtheorem{lemma}[theorem]{Lemma}
\newtheorem*{lemma*}{Lemma}
\newtheorem{corollary}[theorem]{Corollary}
\newtheorem*{conjecture*}{Conjecture}

\newtheorem*{fact*}{Fact}

\newtheorem*{hypothesis*}{Hypothesis}

\theoremstyle{definition}
\newtheorem{definition}[theorem]{Definition}

\theoremstyle{remark}

\newtheorem*{claim*}{Claim}

\newtheorem*{remark*}{Remark}
\newtheorem{observation}[theorem]{Observation}
\newtheorem*{observation*}{Observation}

\usepackage[
letterpaper,
top=1in,
bottom=1in,
left=1in,
right=1in]{geometry}

\usepackage{amsmath,amsfonts,amssymb}

\ifnum\showkeys=1
\usepackage[color]{showkeys}
\fi

\ifnum\showcolorlinks=1
\usepackage[
pagebackref,
colorlinks=true,
urlcolor=blue,
linkcolor=blue,
citecolor=OliveGreen,
]{hyperref}
\fi

\ifnum\showcolorlinks=0
\usepackage[
pagebackref,
colorlinks=false,
pdfborder={0 0 0}
]{hyperref}
\fi

\usepackage{prettyref}

\newcommand{\savehyperref}[2]{\texorpdfstring{\hyperref[#1]{#2}}{#2}}

\newrefformat{eq}{\savehyperref{#1}{\textup{(\ref*{#1})}}}
\newrefformat{eqn}{\savehyperref{#1}{\textup{(\ref*{#1})}}}
\newrefformat{lem}{\savehyperref{#1}{Lemma~\ref*{#1}}}
\newrefformat{def}{\savehyperref{#1}{Definition~\ref*{#1}}}
\newrefformat{thm}{\savehyperref{#1}{Theorem~\ref*{#1}}}
\newrefformat{cor}{\savehyperref{#1}{Corollary~\ref*{#1}}}
\newrefformat{cha}{\savehyperref{#1}{Chapter~\ref*{#1}}}
\newrefformat{sec}{\savehyperref{#1}{Section~\ref*{#1}}}
\newrefformat{app}{\savehyperref{#1}{Appendix~\ref*{#1}}}
\newrefformat{tab}{\savehyperref{#1}{Table~\ref*{#1}}}
\newrefformat{fig}{\savehyperref{#1}{Figure~\ref*{#1}}}
\newrefformat{hyp}{\savehyperref{#1}{Hypothesis~\ref*{#1}}}
\newrefformat{alg}{\savehyperref{#1}{Algorithm~\ref*{#1}}}
\newrefformat{rem}{\savehyperref{#1}{Remark~\ref*{#1}}}
\newrefformat{item}{\savehyperref{#1}{Item~\ref*{#1}}}
\newrefformat{step}{\savehyperref{#1}{step~\ref*{#1}}}
\newrefformat{conj}{\savehyperref{#1}{Conjecture~\ref*{#1}}}
\newrefformat{fact}{\savehyperref{#1}{Fact~\ref*{#1}}}
\newrefformat{prop}{\savehyperref{#1}{Proposition~\ref*{#1}}}
\newrefformat{prob}{\savehyperref{#1}{Problem~\ref*{#1}}}
\newrefformat{claim}{\savehyperref{#1}{Claim~\ref*{#1}}}
\newrefformat{relax}{\savehyperref{#1}{Relaxation~\ref*{#1}}}
\newrefformat{red}{\savehyperref{#1}{Reduction~\ref*{#1}}}
\newrefformat{part}{\savehyperref{#1}{Part~\ref*{#1}}}
\newrefformat{obs}{\savehyperref{#1}{Observation~\ref*{#1}}}
\newrefformat{corr}{\savehyperref{#1}{Corollary~\ref*{#1}}}

\newcommand{\Sref}[1]{\hyperref[#1]{\S\ref*{#1}}}

\usepackage{nicefrac}

\ifnum\usemicrotype=1
\usepackage{microtype}
\fi

\ifnum\showauthornotes=1
\newcommand{\Authornote}[2]{{\sffamily\small\color{red}{[#1: #2]}}}
\newcommand{\Authornotecolored}[3]{{\sffamily\small\color{#1}{[#2: #3]}}}
\newcommand{\Authorcomment}[2]{{\sffamily\small\color{gray}{[#1: #2]}}}
\newcommand{\Authorstartcomment}[1]{\sffamily\small\color{gray}[#1: }

\newcommand{\Authorfnote}[2]{\footnote{\color{red}{#1: #2}}}
\newcommand{\Authorfixme}[1]{\Authornote{#1}{\textbf{??}}}
\newcommand{\Authormarginmark}[1]{\marginpar{\textcolor{red}{\fbox{\Large #1:!}}}}
\else
\newcommand{\Authornote}[2]{}
\newcommand{\Authornotecolored}[3]{}
\newcommand{\Authorcomment}[2]{}
\newcommand{\Authorstartcomment}[1]{}

\newcommand{\Authorfnote}[2]{}
\newcommand{\Authorfixme}[1]{}
\newcommand{\Authormarginmark}[1]{}
\fi

\ifnum\showfixme=0

\fi

\usepackage{boxedminipage}

\newcommand{\Paren}[1]{\left(#1\right)}

\newcommand{\Abs}[1]{\left\lvert#1\right\rvert}

\newcommand{\card}[1]{\lvert#1\rvert}

\newcommand{\norm}[1]{\lVert#1\rVert}

\newcommand{\iprod}[1]{\langle#1\rangle}

\newcommand{\Esymb}{\mathbb{E}}
\newcommand{\Psymb}{\mathbb{P}}

\DeclareMathOperator*{\E}{\Esymb}

\DeclareMathOperator*{\ProbOp}{\Psymb}

\renewcommand{\Pr}{\ProbOp}

\newcommand{\tensor}{\otimes}

\newcommand{\textparen}[1]{\text{(#1)}}

\ifx\because\undefined
\newcommand{\because}[1]{\textparen{because #1}}
\else
\renewcommand{\because}[1]{\textparen{because #1}}
\fi

\newcommand{\defeq}{\stackrel{\mathrm{def}}=}

\newcommand{\mper}{\,.}
\newcommand{\mcom}{\,,}

\newcommand\bdot\bullet

\ifx\mathds\undefined 
\DeclareMathOperator{\Ind}{\mathbb{I}}
\else
\DeclareMathOperator{\Ind}{\mathds 1}}
\fi

\DeclareMathOperator{\Tr}{Tr}

\DeclareMathOperator{\poly}{poly}

\DeclareMathOperator{\polylog}{polylog}

\newcommand{\etal}{et al.\xspace}

\newcommand{\N}{\mathbb N}
\newcommand{\R}{\mathbb R}

\newcommand{\cA}{\mathcal A}

\newcommand{\cC}{\mathcal C}

\newcommand{\cG}{\mathcal G}

\newcommand{\cL}{\mathcal L}

\newcommand{\cN}{\mathcal N}

\newcommand{\cP}{\mathcal P}

\newcommand{\cR}{\mathcal R}

\ifnum\showdraftbox=1
\newcommand{\draftbox}{\begin{center}
  \fbox{%
    \begin{minipage}{2in}%
      \begin{center}%
          \Large\textsc{Working Draft}\\%
        Please do not distribute%
      \end{center}%
    \end{minipage}%
  }%
\end{center}
\vspace{0.2cm}}
\else
\newcommand{\draftbox}{}
\fi

\let\epsilon=\varepsilon

\numberwithin{equation}{section}

\newcommand\MYcurrentlabel{xxx}

\newcommand{\MYstore}[2]{%
  \global\expandafter \def \csname MYMEMORY #1 \endcsname{#2}%
}

\newcommand{\MYload}[1]{%
  \csname MYMEMORY #1 \endcsname%
}

\newcommand{\MYnewlabel}[1]{%
  \renewcommand\MYcurrentlabel{#1}%
  \MYoldlabel{#1}%
}

\newcommand{\MYdummylabel}[1]{}

\newcommand{\torestate}[1]{%
  \let\MYoldlabel\label%
  \let\label\MYnewlabel%
  #1%
  \MYstore{\MYcurrentlabel}{#1}%
  \let\label\MYoldlabel%
}

\newcommand{\restatetheorem}[1]{%
  \let\MYoldlabel\label
  \let\label\MYdummylabel
  \begin{theorem*}[Restatement of \prettyref{#1}]
    \MYload{#1}
  \end{theorem*}
  \let\label\MYoldlabel
}

\newcommand{\restatelemma}[1]{%
  \let\MYoldlabel\label
  \let\label\MYdummylabel
  \begin{lemma*}[Restatement of \prettyref{#1}]
    \MYload{#1}
  \end{lemma*}
  \let\label\MYoldlabel
}

\newcommand{\restateprop}[1]{%
  \let\MYoldlabel\label
  \let\label\MYdummylabel
  \begin{proposition*}[Restatement of \prettyref{#1}]
    \MYload{#1}
  \end{proposition*}
  \let\label\MYoldlabel
}

\newcommand{\restatefact}[1]{%
  \let\MYoldlabel\label
  \let\label\MYdummylabel
  \begin{fact*}[Restatement of \prettyref{#1}]
    \MYload{#1}
  \end{fact*}
  \let\label\MYoldlabel
}

\newcommand{\restate}[1]{%
  \let\MYoldlabel\label
  \let\label\MYdummylabel
  \MYload{#1}
  \let\label\MYoldlabel
}

\newcommand{\addreferencesection}{
  \phantomsection
  \addcontentsline{toc}{section}{References}
}

\newcommand{\sse}{\subseteq}

\newcommand{\e}{\epsilon}
\newcommand{\eps}{\epsilon}

\let\origparagraph\paragraph
\renewcommand{\paragraph}[1]{\origparagraph{#1.}}

\allowdisplaybreaks

\newcommand{\cclassmacro}[1]{\texorpdfstring{\textbf{#1}}{#1}\xspace}

\newcommand{\np}{\cclassmacro{NP}}

\usepackage{paralist}

\usepackage{comment}

\let\citet\cite

\theoremstyle{definition}

\DeclareUrlCommand\email{}

\usepackage{tikz,bbm}
\newcommand{\Span}{\mathop{\mathrm{span}}}

\newcommand{\col}{\mathop{\mathrm{col}}}
\newcommand{\diag}{\mathop{\mathrm{diag}}}
\newcommand{\Id}{\mathrm{Id}}

\newcommand{\pref}{\prettyref}

\newcommand{\pE}{{\mathbb{\tilde E}}}

\newcommand{\barn}{\overline{n}}
\newcommand{\ovec}{{\mathbbm 1}}
\newcommand{\novec}{\tilde\ovec}
\newcommand{\sdpval}{\mathsf{sdpval}}

\title{Tight Lower Bounds for Planted Clique \\
in the Degree-4 SOS Program}

\author{%
\normalsize
Prasad Raghavendra \thanks{UC Berkeley,
  \protect\email{prasad@cs.berkeley.edu}. Supported by NSF Career
  Award, NSF CCF-1407779 and the Alfred. P. Sloan Fellowship. }
\and
\normalsize
Tselil Schramm\thanks{UC Berkeley, \protect\email{tschramm@cs.berkeley.edu}.
Supported by an NSF Graduate Research Fellowship (NSF award no 1106400).}
}

\date{}

\begin{document}

\maketitle

\draftbox

\thispagestyle{empty}

\begin{abstract}
    We give a lower bound of $\tilde{\Omega}(\sqrt{n})$ for the degree-4 Sum-of-Squares SDP relaxation for the planted clique problem.
Specifically, we show that on an Erd\"{o}s-R\'{e}nyi graph $G(n,\tfrac{1}{2})$, with high probability there is a feasible point for the degree-4 SOS relaxation of the clique problem with an objective value of $\tilde{\Omega}(\sqrt{n})$,
so that the program cannot distinguish between a random graph and a random graph with a planted clique of size $\tilde{O}(\sqrt{n})$.
This bound is tight.

We build on the works of Deshpande and Montanari and Meka et al., who give lower bounds of $\tilde{\Omega}(n^{1/3})$ and $\tilde{\Omega}(n^{1/4})$ respectively.
We improve on their results by making a perturbation to the SDP solution proposed in their work, then showing that this perturbation remains PSD as the objective value approaches $\tilde{\Omega}(n^{1/2})$.

In an independent work, Hopkins, Kothari and Potechin \cite{HopkinsKP15} have obtained
a similar lower bound for the degree-$4$ SOS relaxation.

\end{abstract}

\clearpage

\section{Introduction}

In the {Maximum Clique} problem, the input consists of a graph $G
= (V,E)$ and the goal is to find the largest subset $S$ of vertices
all of which are connected to each other.  The Maximum Clique problem
is \np-hard to approximate within a $n^{1-\e}$-factor for all $\e > 0$ \cite{Hastad96,Khot01}.

Karp \cite{Karp76} suggested an average case version of the
Maximum Clique problem on random graphs drawn from the  Erd\"{o}s-R\'{e}nyi distribution $\mathbb{G}(n,\tfrac{1}{2})$.
A heuristic argument shows that an Erd\"{o}s-R\'{e}nyi graph $G \sim \mathbb{G}(n,\tfrac{1}{2})$ has a clique of size $(1-o(1))\log n$ with high probability: given such a graph, choose a random vertex, then choose one of its neighbors, then choose a vertex adjacent to both, and continue this process until there is no vertex adjacent to the clique.
After $\log n$ steps, the probability that another vertex can be added is $\tfrac{1}{n}$, and so after about $\log n$ steps this process terminates.
 This heuristic argument can be made precise, and one can show that this greedy algorithm can find a clique of size $(1 + o(1))\log n$ in an instance of $\mathbb{G}(n,\tfrac{1}{2})$ in polynomial time.

Indeed, with some work it can be shown that the largest clique in an instance of $\mathbb{G}(n,\tfrac{1}{2})$ actually has size $(2 \pm o(1))\log n$ with high probability \cite{GrimmetM75,Matula76,BollobasE76}.
But while some clique of size $(1 \pm o(1))\log n$ can easily be found in polynomial time (using the heuristic from the previous paragraph), an efficient algorithm for finding the clique of size $2\log n$ has been much more elusive.
In his seminal paper on the probabilistic analysis of combinatorial algorithms, Karp asked whether there exists a polynomial-time algorithm for finding a clique of  size $(1+\epsilon)\log n$ for any fixed constant $\epsilon > 0$ \cite{Karp76}.
Despite extensive efforts, there has been no algorithmic progress on
this question since.

The {\it planted clique problem} is a natural variant of this problem wherein the
input is promised to be either a graph drawn from $G
\sim \mathbb{G}(n,\frac{1}{2})$ or a graph $G \sim \mathbb{G}(n,\frac{1}{2})$ with a
clique of size $k$ planted within its vertices.  The goal of the
algorithm is to distinguish between the two distributions.

For $k > (2+\e) \log n$, there is a simple quasi-polynomial time
algorithm that distinguishes the two distributions.  The algorithm
simply tries all subsets of $(2 + \epsilon)\log n$ vertices, looking
for a clique.  For a random graph $\mathbb{G}(n,\frac{1}{2})$, there are
no cliques of size $(2+\e) \log n$, but there is one in the planted
distribution.  Clearly, the planted clique problem becomes easier as the planted clique's size $k$ increases.
Yet there are no polynomial-time algorithms known for
this problem for any $k < o(\sqrt{n})$.  For $k = \Omega(\sqrt{n})$, a result of
Alon et al. uses random matrix theory to argue that looking at the
spectrum of the adjacency matrix suffices to solve the decision
problem \cite{AlonKS98}.

The works of \cite{FriezeK08,BrubakerV09} show that, if one were able to efficiently calculate the injective tensor norm of a certain random order-$m$ tensor, then by extending the spectral algorithm of \cite{AlonKS98} one would have a polynomial-time algorithm for $k > n^{1/m}$.
However, there is no known algorithm that efficiently computes the injective tensor norm of an order-$m$ tensor; in fact computing the inective tensor norm is hard to approximate in the general case \cite{HarrowM13}.

While algorithmic progress has been slow, there has been success in proving strong lower bounds for the planted clique problem within specific algorithmic frameworks.
The first such bound was given by Jerrum, who showed that a class of Markov Chain Monte Carlo algorithms require a super-polynomial number of steps to find a clique of size $(1+\eps)\log n$, for any fixed $\epsilon > 0$, in an instance of $\mathbb{G}(n,\tfrac{1}{2})$ \cite{Jerrum92}.
Feige and Krauthgamer showed that $r$-levels of the Lov\'{a}sz-Schriver SDP hierarchy are needed to find a hidden clique of size $k\ge \tilde{\Omega}(\sqrt{n}/2^r)$ \cite{FeigeK00, FeigeK03}.
Feldman et al. show (for the planted bipartite clique problem) that any ``statistical algorithm'' cannot distinguish in a polynomial number of queries between the random and planted cases for $k < \tilde{O}(\sqrt{n})$ \cite{FeldmanGRVX12}.

More recently, there has been an effort to replicate the results of \cite{FeigeK00,FeigeK03} for the Sum-of-Squares (or SOS) hierarchy, a more powerful SDP hierarchy.
The recent work of \cite{MekaPW15} achieves a $\tilde{\Omega}(n^{1/2r})$-lower bound for $r$-rounds of the SOS hierarchy, by demonstrating a feasible solution for the level-$r$ SDP relaxation with a large enough objective value in the random case.
The work of \cite{DeshpandeM15} achieves a sharper $\tilde{\Omega}(n^{1/3})$ lower bound for the Meka-Potechin-Wigderson SDP solution, but only for $r = 2$ rounds; a counterexample of Kelner (which may be found in \cite{Barak14}) demonstrates that the analysis of \cite{DeshpandeM15} is tight for the integrality gap instance of \cite{DeshpandeM15,MekaPW15} within logarithmic factors.

This line of work brings to the fore the question: can a $d = O(1)$-degree
SOS relaxation solve the { \it planted clique} problem for some $k < \sqrt{n}$?
While lower bounds are known for Lov\'{a}sz-Schrijver SDP relaxations
for planted clique \cite{FeigeK00, FeigeK03}, SOS relaxations can in general be
much more powerful than Lov\'{a}sz-Schrijver relaxations.  For example,
while there are instances of unique games that are hard for
$\poly(\log \log n)$-rounds of the Lov\'{a}sz-Schrijver SDP hierarchy
\cite{KhotS09, RaghavendraS09},
recent work has shown that these instances are solved by degree-$8$
SOS hierarchy \cite{BarakBHKSZ12}.

Moreover, even the degree-$4$ SOS relaxation proves to be surprisingly
powerful in a few applications:
\begin{itemize}
\item First, the work of Barak \etal
\cite{BarakBHKSZ12} shows that a degree $4$ SOS relaxation can certify
$2-to-4$ hypercontractivity of low degree polynomials over the
hypercube.  This argument is the reason that hard
instances for Lov\'{a}sz-Schriver and other SDP hierarchies constructed
via the {\it noisy hypercube gadgets} are easily refuted by the SOS
hierarchy.

\item
Second, a degree-$4$ SOS relaxation can certify
that the $2$-to-$4$ norm of a random subspace of dimension at most
$o(\sqrt{n})$ is bounded by a constant (with high probability over the
choice of the subspace) \cite{BarakBHKSZ12}.
    This average-case problem has superficial similarities to the planted clique problem.
\end{itemize}

In this work, we make modest progress towards a lower bound for SOS
relaxations of planted clique by obtaining a nearly tight lower bound
for the degree-$4$ SOS relaxation (corresponding to two rounds, $r = 2$).
More precisely, our main result is the following.
\begin{theorem}\label{thm:main-result}
    Suppose that $G \sim \mathbb{G}(n,\tfrac{1}{2})$.
    Then with probability $1 - O(n^{-4})$, there exists a feasible solution to the SOS-SDP of degree $d = 4$ ($r = 2$) with objective value $\frac{\sqrt{n}}{\polylog n}$.\footnote{We have made no effort to optimize logarithmic factors in this work; a more delicate analysis of the required logarithmic factors is certainly possible.}
\end{theorem}
Note that by the work of \cite{AlonKS98}, this result is tight up to
logarithmic factors.  In an independent work, Hopkins, Kothari and
Potechin \cite{HopkinsKP15} have obtained a similar result.

Our work builds heavily on previous work by Meka, Potechin and
Wigderson \cite{MekaPW15} and Deshpande and Montanari \cite{DeshpandeM15}.
Since the SDP solution constructed in these works is infeasible for $k
> n^{1/3}$, we introduce a modified SDP solution with objective value $\tilde{\Omega}(\sqrt{n})$, and prove that for a random graph $G$ the solution is feasible with high probability.
At the parameter setting for which the objective value becomes $\Omega(n^{1/3})$, the SDP solutions of \cite{DeshpandeM15,MekaPW15} violate the PSDness constraint, or equivalently, there exists a set of test vectors $X$ such that $x^T M x < 0$ for all $x \in X$.
Our feasible SDP solution is a perturbation of their solution--we add spectral mass to the solution along the vectors from the set $X$, then enforce the linear constraints of the SDP program.

\subsection{Notation}
We use the symbol $\succeq$ to denote the PSD ordering on matrices, saying that $A \succeq 0$ if $A$ is PSD and that $A \succeq B$ if $A- B\succeq 0$.
When we wish to hide constant factors for clarity, we use $a \lesssim b$ to denote that $a \le C\cdot b$ for some constant $C$.

We denote by $\ovec_n \in \R^n$ the vector such that $\ovec_n (i) \defeq 1~\forall i \in [n]$, or the all-1's vector.
We denote the normalized version of this vector by $\novec_n \defeq\ovec_n/\|\ovec_n\| $.
Further, we use $J_n \defeq \ovec\ovec^\top$ and $Q_n \defeq \novec\novec^\top $.
We will drop the subscript when $n$ is clear from context.

In our notation, we at times closely follow the notation of \cite{DeshpandeM15}, as our paper builds on their results and we recycle many of their bounds.

For convenience, we will use the shorthand $\barn = n \log n$.
We will abuse notation by using $\binom{n}{2}$ to refer to both the binomial coefficient and to the set $\binom{n}{2} = \{(a,b)~|~ a,b \in [n],~ a\neq b\}$.
We will also use the notation $\binom{n}{\le k}$ to refer to the union of sets $\bigcup_{i=0}^k \binom{n}{i}$.
Further, when we give a vector $v \in \R^{\binom{n}{2}}$, we will identify the entries of $v$ by unordered pairs of elements of $[n]$.

Throughout the paper, we will (unless otherwise stated) work with some fixed instance $G$ of $\mathbb{G}(n,\tfrac{1}{2})$, and denote by $A_i \in \R^n$ the ``centered'' $i$th row of the adjacency matrix of $G$, with $j$th entry equal to $1$ if the edge $(i,j) \in E$, equal to $-1$ if the edge $(i,j) \not\in E$, and equal to $0$ for $j = i$.
We will use $A_{ij}$ to denote the $j$th index of $A_i$.

\subsection{Organization}

In \pref{sec:overview}, we give background material on the degree-4 SOS relaxation for the max-clique problem, describe the integrality gap of Deshpande and Montanari for the planted clique problem, and explain the obstacle they face to reach an integrality gap value of $\tilde\Omega(\sqrt{n})$.
We then describe our integrality gap instance, motivating our construction using the obstacle for the Deshpande-Montanari and Meka-Potechin-Wigderson witness, and give an overview of our proof that our integrality gap instance is feasible.
In \pref{sec:mainproof},  we prove that our witness is PSD, completing the proof of feasibility.
\pref{sec:matrix-conc} contains our concentration bounds for random matrices that arise within our proofs.
In our proof, we reuse several bounds proved by Deshpande and Montanari.
As far as possible, we restate the claims from
\cite{DeshpandeM15} as they are used; for convenience, in
\pref{app:dm-bounds}, we list a few other claims from Deshpande and Montanari that we use in this paper.

\section{Preliminaries and Proof Overview}\label{sec:overview}
In this section, we describe the degree-4 SOS relaxation for the max-clique SDP and give background on the Deshpande-Montanari witness.
We then describe our own modified witness, and give an overview of the proof that our witness is feasible (the difficult part being showing that our witness is PSD).
The full proof of feasibility is deferred to \pref{sec:mainproof}.

\subsection{Degree-4 SOS Relaxation for Max Clique} \label{sec:sdprel}

The degree $d = 4$ SOS relaxation for the maximum clique problem is
a semidefinite program whose variables are $X \in \R^{\binom{n}{\leq 2} \times
  \binom{n}{\leq 2}}$.  For a subset $S \sse V$ with $|S| \leq 2$, the
variable $X_S$ indicates whether $S$ is contained in the maximum
clique.  For a graph $G$ on $n$ vertices, the program can be described as follows.
\begin{align}
  \text{Maximize  } & \sum_{i \in [n]} X_{\{i\}, \{i\}} \label{eq:sdprelaxation} \\
\text{ subject to } & X_{S_1,S_2} \in [0,1] & \forall S_1, S_2 \in \binom{n}{\leq 2} \nonumber\\
& X_{S_1,S_2} = X_{S_3,S_4} & \text{ whenever } S_1 \cup S_2 = S_3 \cup S_4 \nonumber\\
& X_{S_1, S_2 } = 0 & \text{ if } S_1 \cup S_2 \text{ is not a clique in } G
                    \nonumber \\
& X_{\emptyset, \emptyset} = 1 \nonumber\\
& X \succeq 0 \nonumber
\end{align}
It is instructive to think of the variable $X_{S}$ as a \emph{pseudoexpectation} of the product of indicator variables, or a \emph{pseudomoment}:
\[
    X_{S} = \pE\left[\prod_{i \in S} \Ind(i \in \text{clique} )\right].
\]
Intuitively, the constraints of the SDP force the solution to behave somewhat like the moments of a probability distribution over integral solutions, although they needn't correspond to the moments of a true distribution, hence the term \emph{pseudomoment}.
For more background, see e.g. \cite{Barak14}.
The pseudmoment interpretation of the SDP solution motivates the choice of the witness in the prior work.
For example, we may notice that the objective function in this view is simply the pseudoexpectation of the size of the planted clique, $\pE[\sum_{i\in[n]} \Ind(i \in \text{clique})]$.

If $\sdpval(G,4)$ denotes the optimum value of the SDP relaxation on
graph $G$, then clearly $\sdpval(G)$ is at least the size of the
maximum clique in $G$.
In order to prove a lower bound for degree $4$ SOS relaxation on
$\mathbb{G}(n,\frac{1}{2})$, it is sufficient to argue that with overwhelming
probability, $\sdpval(G)$ is significantly larger than the maximum
clique on a random graph.
This amounts to exhibiting a feasible SDP
solution with large objective value, for an overwhelming fraction of graphs sampled from $\mathbb{G}(n,\frac{1}{2})$.
Formally, we will show the following:
\begin{theorem}[Formal version of \pref{thm:main-result}]\label{thm:main-technical}
There exists an absolute constant $c \in \N$ such that
  $$ \Pr_{G \sim \mathbb{G}(n,\frac{1}{2}) } \left\{ \sdpval(G) \geq
  \frac{\sqrt{n}}{\log^c{n}} \right\} \geq 1-O(n^{-4})$$
\end{theorem}

We obtain \pref{thm:main-technical} by constructing a point, or witness, for each $G \sim G(n,\tfrac{1}{2})$, then proving that the point is feasible with high probability.
We defer the description of our witness to \pref{def:sol} and \pref{def:sol-pE}, as we spend \pref{sec:DM-wit} and \pref{sec:prob-space} motivating our construction; however the curious reader may skip ahead to \pref{def:sol-pE} which does not require the knowledge of additional notation.

\subsection{Deshpande-Montanari Witness} \label{sec:DM-wit}

Henceforth, fix a graph $G$ that is sampled from $\mathbb{G}(n,\frac{1}{2})$.
Both the work of Meka, Potechin and Wigderson \cite{MekaPW15} and
that of Deshpande and Montanari \cite{DeshpandeM15} construct
essentially the same SDP solution for the degree-$4$ SOS relaxation.

This SDP solution assigns to each clique of size $1,\ldots, d$, a value
that depends only on its size (in our case, $d = 4$).
In essence, their solution takes advantage of the independence of the $G(n,p)$ instance.
The motivating observation is that the variable $X_S$ can be thought of as a pseudoexpectation of the indicator that $S$ is a subclique of the planted clique. The idea is then to make this pseudoexpectation of the indicator consistent with the true expectation under the distribution where a clique of size $k$ is planted uniformly at random within the instance of $G(n,p)$.
Thus, every vertex is in the clique ``with uniform probability:''
\[
    \pE[X_{\{i\}}]
    \approx \E[\Ind(i \text{ is in planted clique})]
    = \frac{k}{n}.
\]
Then, the same principle is applied to edges, traingles, and $4$-cliques,
so that
\[
    \pE[X_{S}]
    \approx \Ind(S \text{ is clique}) \cdot \E[\Ind(S \text{ is in planted clique})]
    =  \frac{\binom{k}{|S|}}{\binom{n}{|S|}}\cdot \Paren{\tfrac{1}{2}}^{\binom{|S|}{2}}.
\]

This is the general idea of the SDP solution of \cite{DeshpandeM15}.
More formally, the SDP solution in \cite{DeshpandeM15} is specified by
four parameters $\underline{\alpha} = \{\alpha_i\}_{i=1}^4$ as,
\begin{align*}
  M(G,\underline{\alpha}) = \alpha_{\card{A\cup B}} \cdot \cG_{A \cup B} \mcom
\end{align*}
where for a set of vertices $A \sse V$, $\cG_{A}$ is the indicator that
the subgraph induced on $A$ is a clique.  The parameters $\{\alpha_i\}_{i\in
  [4]}$ determine the value of the objective function, and the
feasibility of the solution.  As a convention, we will define $\alpha_0 = 1$.

  It is easy to check that the solution $M(G,\underline{\alpha})$ satisfies
all the linear constraints of the SOS program \eqref{eq:sdprelaxation}, since it assigns
non-zero values only to cliques in $G$.  The key difficulty is in
showing that the matrix $M$ is PSD for an appropriate choice of
parameters $\underline{\alpha}$.

In order to show that $M(G,\underline{\alpha}) \succeq 0$, it is sufficient
to show that $N(G,\underline{\alpha}) \succeq 0$ where,
\begin{align*}
N_{A,B} = \alpha_{\card{A \cup B}} \cdot \prod_{i \in A \setminus B, j \in
  B \setminus A} \cG_{ij},
\end{align*}
where $\cG_{ij}$ is the indicator for the presence of the edge $(i,j)$.
In words, $N$ is the matrix where the entry $\{a,b,c,d\}$ is
proportional not to the indicator of whether $\{a,b,c,d\}$ is a clique,
but to the indicator of whether $G$ has as a subgraph the bipartite clique with
bipartitions $\{a,b\}$ and $\{c,d\}$.  It is easy to see that the matrix
$M$ is obtained by dropping from $N$ the rows and columns
corresponding $\{a,b\} \in \binom{n}{2}$ where $(a,b) \notin E(G)$.  Hence $N
\succeq 0 \Longrightarrow M \succeq 0$.

Notice that $N$ is a random matrix whose entries
depend on the edges in the random graph $G$.  At the risk of
over-simplification, the approach of both the previous works
\cite{MekaPW15} and \cite{DeshpandeM15} can be broadly summarized as follows:
\begin{enumerate}
\item (Expectation) Show that the expected matrix $\E[N]$ has sufficiently large
positive eigenvalues.
\item (Concentration) Show that with high probability over the choice
  of $G$, the {\it noise} matrix $N - \E[N]$ has bounded eigenvalues, so as to ensure that $N = \E[N] + (N-\E[N]) \succeq
0$
\end{enumerate}

Here we will sketch a few key details of the argument in \cite{DeshpandeM15}.
The matrix $N\in  \R^{\binom{n}{\leq 2} \times \binom{n}{\leq 2}}$ can be
decomposed into blocks $\{N_{ab}\}_{a,b \in \{0,1,2\}}$ where $N_{a,b}
\in \R^{\binom{n}{a} \times \binom{n}{b}}$.  Deshpande and Montanari use
the Schur complements to reduce the problem of proving that $N \succeq
0$ to facts about the blocks $\{N_{ab}\}_{a,b \in \{0,1,2\}}$.
Specifically, they show the following lemma:
\begin{lemma}\label{lem:schur}
    Let $\cA \in \R^{\binom{n}{\leq 2} \times \binom{n}{\leq 2}}$ be the matrix defined so that $\cA_{A,B} = \alpha_{|A|}\alpha_{|B|}$.
    For $a,b\in\{0,1,2\}$, let $H_{a,b}$ be the submatrix of $N(G,\alpha) - \cA$ corresponding to monomials $X_S$ with $|S| = a+b$.
    Then $N(G,\alpha)$ is PSD if and only if
    \begin{align}
	H_{11} &\succeq 0,\label{eq:one-one}\\
	H_{22} - H_{12}^TH_{11}^{-1}H_{12} &\succeq 0 \label{eq:comp}
    \end{align}
\end{lemma}

The most significant challenge is to argue that \eqref{eq:comp} holds
with high probability.  In fact, the inequality only holds for the Deshpande-Montanari SDP solution with high
probability for parameters $\alpha$ for which the objective value is $o(n^{1/3})$.

\paragraph{Expected matrix}
The expected matrix $\E[H_{22}]$ is symmetric with respect to
permutations of the vertices.  It forms an {\it association scheme}
(see \cite{MekaPW15,DeshpandeM15}), by
virtue of which its eigenvalues and eigenspaces are well understood.
In particular, the following proposition in \cite{DeshpandeM15} is an
immediate consequence of the theory of association schemes.

    \begin{proposition}[Proposition 4.16 in \cite{DeshpandeM15}]\label{prop:eigenspaces}
   $\E[H_{22}]$ has three eigenspaces, $V_0,V_1,V_2$ such that
    \[
	\E[H_{22}] = \lambda_0 \Pi_{0} + \lambda_1 \Pi_1 + \lambda_2 \Pi_2,
    \]
    where $\Pi_0,\Pi_1,\Pi_2$ are the projections to the spaces
    $V_0,V_1,V_2$ respectively.  The eigenvalues are given by,
\begin{align}
\lambda_0 (\underline{\alpha}) & \defeq \alpha_2 + (n-2)\alpha_3 + \frac{(n-2)(n-3)}{32} \cdot \alpha_4 - \frac{n(n-1)}{2}
     \alpha_2^2 \label{eq:lambda0} \\
\lambda_1 (\underline{\alpha}) & \defeq \alpha_2 + \frac{(n-4)}{2} \alpha_3 - \frac{(n-3)}{16} \alpha_4 \label{eq:lambda1} \\
\lambda_2 (\underline{\alpha}) & \defeq \alpha_2 - \alpha_3 + \frac{\alpha_4}{16}  \label{eq:lambda2}
\end{align}
    Further the eigenspaces are given by,
    \begin{align*}
	V_0
	&= \Span\{\ovec\},\\
	V_{1}
	&= \Span\{ u~|~\iprod{u,\ovec} = 0, ~u_{i,j} = x_i + x_j~\text{for}~x \in \R^{n}\},\\
	\text{and} \qquad V_2 &= R^{\binom{n}{\le 2}} \setminus (V_0 \cup V_1)\mcom
    \end{align*}
    where we have used $\R^{\binom{n}{\le 2}}$ to denote the space of vectors of real numbers indexed by subsets of $n$ of size at most $2$.

\end{proposition}

\paragraph{Deviation from Expectation}
Given the lower bound on eigenvalues of the expected matrix
$\E[H_{22}]$, the next step would be to bound the spectral norm of the
noise $H_{22} - \E[H_{22}]$.
However, since the eigenspaces of $\E[H_{22}]$ are stratified (for the given $\underline\alpha$), with
one large eigenvalue and several much smaller eigenvalues, standard
matrix concentration does not suffice to give tight bounds.  To overcome this, Deshpande and Montanari split $H_{22}$ and $H_{12}^TH_{11}^{-1}H_{12}$ along the eigenspaces of $\E[H_{22}]$.

More precisely, let us split $H_{22} - \E[H_{22}]$ as
$$H_{22} - \E[H_{22}] = Q + K$$
where $Q$ includes all multilinear entries, and $K$ includes all
non-multilinear entries, i.e., entries $K(A,B)$ where $A \cap B \neq \emptyset$.
Formally,
$$ Q(A,B) = \begin{cases} H_{22}(A,B) - \E[H_{22}](A,B) & \text{ if } A
  \cap B = \emptyset\\ 0 & \text{ otherwise } \end{cases}\mper$$
The spectral norm of the matrix $Q$ over the eigenspaces $V_0,V_1,V_2$
is carefully bounded in \cite{DeshpandeM15}.
\begin{lemma} (Proposition 4.20, 4.25 in \cite{DeshpandeM15}) \label{lem:qkbounds}
With probability at least $1- O(n^{-4})$, all of the following bounds
hold:
\begin{align}
\norm{\Pi_a Q \Pi_b} & \lesssim \alpha_4 \barn^{3/2} \qquad \forall (a,b) \in
                  \{0,1,2\}^2 \label{eq:qabub} \\
\norm{\Pi_2 Q \Pi_2} & \lesssim \alpha_4 \barn  \label{eq:q22ub}\\
\norm{K} & \lesssim \alpha_3 \barn^{1/2} \label{eq:kub}
\end{align}
\end{lemma}
\pref{prop:eigenspaces} and \pref{lem:qkbounds} are sufficient to
conclude that $H_{22} \succeq 0$ for parameter choices of $\alpha$ that
correspond to planted clique of size up to $\omega(n^{1/3})$.  More
precisely, to argue that with high probability $H_{22} \succeq 0$, it
is sufficient to argue that, $\E[H_{22}] \succeq \E[H_{22}] -
H_{22}$, i.e.,
\begin{align*}
\begin{bmatrix}
\lambda_0 & 0 & 0\\
0 & \lambda_1 & 0\\
0 & 0 & \lambda_2
\end{bmatrix} \succeq
\begin{bmatrix}
\norm{\Pi_0 Q \Pi_0} & \norm{\Pi_0 Q \Pi_1} & \norm{\Pi_0 Q \Pi_2}\\
\norm{\Pi_1Q \Pi_0} & \norm{\Pi_1 Q \Pi_1} & \norm{\Pi_1 Q \Pi_2}\\
\norm{\Pi_2 Q \Pi_0} & \norm{\Pi_2 Q \Pi_1} & \norm{\Pi_2 Q \Pi_2}
\end{bmatrix} + \alpha_3 \barn^{1/2} \cdot \Id\mper\footnotemark
\end{align*}
\footnotetext{
    Here, we have identified the matrices $\E[H_{22}]$ and $\E[H_{22}] - H_{22}$, which are matrices in $\R^{\binom{n}{\le 2}}$, with the $3 \times 3$ matrices corresponding to diagonalizing $\E[H_{22}]$ according to the three eigenspaces $V_0,V_1,V_2$ of the expectation $\E[H_{22}]$.
    This is analagous to decomposing any quadratic form $v^\top H_{22} v$ into $v^\top (\Pi_0 + \Pi_1 + \Pi_2)H_{22}(\Pi_0 + \Pi_1 + \Pi_2)v$.
}

Deshpande and Montanari fix $\alpha_1 = \kappa$, $\alpha_2 = 4\kappa^2$, $\alpha_3 = 8 \kappa^3$ and $\alpha_4
= 512 \kappa^4$ for a parameter $\kappa$. Using \pref{prop:eigenspaces} and
\pref{lem:qkbounds}, the above matrix inequality becomes,
\begin{align}
\begin{bmatrix}
n^2 \kappa^4 & 0 & 0\\
0 & n \kappa^3 & 0\\
0 & 0 & \kappa^2
\end{bmatrix} \succeq \kappa^4
\begin{bmatrix}
n^{3/2} & n^{3/2} & n^{3/2}\\
n^{3/2} & n^{3/2} & n^{3/2}\\
n^{3/2} & n^{3/2} & n
\end{bmatrix} \mcom \label{eq:dmh22psd}
\end{align}
which can be shown to hold for $\kappa \ll n^{-2/3}$.  Eventually, it is
necessary to show \eqref{eq:comp}, which is stronger than $H_{22}
\succeq 0$.  This is again achieved by showing bounds on the spectra of $H_{11}^{-1}$
 and $H_{12}$.  We refer the reader to \cite{DeshpandeM15} for more
 details of the arguments.

 \subsection{Problematic Subspace}\label{sec:prob-space}

The SDP solution described above ceases to be PSD at $\kappa \simeq n^{-2/3}$
which corresponds to an objective value of $O(n^{1/3})$.  The specific
obstruction to $H_{22} \succeq 0$ arises out of \eqref{eq:dmh22psd}.
More precisely, the bottom $2 \times 2$ principal minor which yields the
constraint,
\begin{align*}
\begin{bmatrix}
\lambda_1 & \norm{\Pi_1 Q \Pi_2}\\
\norm{\Pi_2 Q\Pi_1} & \lambda_2
\end{bmatrix} \approx
\begin{bmatrix}
n \kappa^3 & -n^{3/2} \kappa^4\\
-n^{3/2} \kappa^4 & \kappa^2
\end{bmatrix} \succeq 0
\end{align*}
forcing $\kappa \ll n^{-2/3}$.  It is clear that the problematic vectors $x
\in \R^{\binom{n}{2}}$ for which $x^T H_{22} x < 0$ are precisely those
 for which $x^T \Pi_2 Q \Pi_1 x < 0$ and $|x^T \Pi_2 Q \Pi_1 x|$ is large, i.e.,
 $\Pi_2 x$ aligns with the subspace $Q (V_1 \oplus V_0)$.

In fact, we identify a specific subspace $W$ that is problematic for
the \cite{DeshpandeM15} solution.  To describe the subspace, let us fix some notation. Define the random variable $A_{ij}$ to be $-1$ if $(i,j) \not\in E$, and $+1$ otherwise.
We follow the convention that $A_{ii} = 0$.
\
\begin{lemma} \label{lem:probsubspace}

Let the vectors $a_1,\ldots,a_n \in \R^{\binom{n}{2}}$ be defined so that $
a_{i}(k,\ell) \defeq A_{ik}A_{i\ell}$, and let $W \defeq \Span\{a_1,\ldots,a_n\}$.  Then with
probability at least $1- O(n^{-4})$,
$$ \norm{\Pi_2 Q  - \Pi_2 \Pi_W Q} \lesssim \alpha_4 \barn$$
\end{lemma}
\begin{proof}
This is an immediate observation from the various matrix norm bounds
in \cite{DeshpandeM15} (specifically \pref{lem:tildes}, \pref{lem:wigner} and \pref{obs:nullspace}).
We defer the detailed proof to \pref{app:misc}.
\end{proof}
Since $\norm{\Pi_2 Q \Pi_1} \gg \alpha_4 \barn$, the above lemma implies
that all the vectors with large singular values for $Q$ are
within the subspace $W$.  Furthermore, we will show the following
lemma which clearly articulates that $W$ is the sole obstruction to $H_{22}
\succeq 0$.

\begin{lemma}\label{lem:calcs}
Suppose $\underline{\alpha} \in \R_+^{4}$ satisfies
\begin{align}
  \min(\lambda_0(\underline{\alpha}), \lambda_1(\underline{\alpha}),
  \lambda_2(\underline{\alpha})) & \gg \alpha_3 \barn^{1/2} \mcom \label{eq:cond1} \\
\lambda_0(\underline{\alpha})  > \lambda_1(\underline{\alpha}) &\gg \alpha_4
  \barn^{3/2} \mcom\label{eq:cond2}\\
  \lambda_2(\underline{\alpha}) & \gg \alpha_4 \barn \label{eq:cond3}
\end{align}
then with probability $1-O(n^{-4})$,
    \[
	 H_{22}
\succeq
\frac{1}{4} \cdot \E[H_{22}]
- \frac{16\norm{Q}^2}{\lambda_1}\cdot \Pi_2 \Pi_W \Pi_2.
    \]
\end{lemma}

\begin{proof}  Fix $\theta = \frac{16 \norm{Q}^2}{\lambda_1}$.
    Recall that $H_{22} - \E[H_{22}] = Q + K$.
We can write the matrix $$H_{22}+ \theta \cdot \Pi_2
\Pi_W \Pi_2  = B_{W^{\perp}} + B_W + B_K + \frac{1}{4} \E[H_{22}]\mcom$$
where
\begin{align*}
  B_{W^{\perp}} = \frac{1}{4} \E[H_{22}] + \begin{bmatrix}
 \Pi_0 Q \Pi_0 & \Pi_0 Q \Pi_1 & \Pi_0 Q (I-\Pi_W)\Pi_2\\
\Pi_1Q \Pi_0  & \Pi_1 Q \Pi_1 & \Pi_1 Q (I-\Pi_W) \Pi_2\\
\Pi_2 (I-\Pi_W) Q \Pi_0 & \Pi_2 (I-\Pi_W) Q \Pi_1 & \Pi_2 Q \Pi_2
\end{bmatrix}\footnotemark
\end{align*}
\footnotetext{Here again we diagonalize according to the subspaces $V_0,V_1,V_2$, as in \pref{eq:dmh22psd}}
and
\begin{align*}
 B_W = \frac{1}{4} \E[H_{22}] + \begin{bmatrix}
 0 & 0 & \Pi_0 Q \Pi_W \Pi_2\\
0  & 0 & \Pi_1 Q \Pi_W \Pi_2 \\
\Pi_2 \Pi_W Q \Pi_0 & \Pi_2 \Pi_W Q \Pi_1 & \theta \cdot \Pi_2 \Pi_W \Pi_2
\end{bmatrix}
\end{align*}
and $B_K = K + \frac{1}{4} \E[H_{22}]$.

It is sufficient to show that $B_{W^{\perp}}, B_W$ and $B_K \succeq
0$.  Using \pref{prop:eigenspaces} and \eqref{eq:kub}, $B_K \succeq
 (\lambda_0-\alpha_3 \barn^{1/2}) \Pi_0 + (\lambda_1 - \alpha_3 \barn^{1/2}) \Pi_1
 + (\lambda_2 - \alpha_3 \barn^{1/2})\Pi_2 \succeq 0$ when condidition \pref{eq:cond1} holds.
Using \pref{prop:eigenspaces}, \pref{lem:qkbounds} and
\pref{lem:probsubspace} we can write,

\begin{align*}
 B_{W^{\perp}} \succeq \frac{1}{4} \cdot \begin{bmatrix}
\lambda_0 & 0 & 0\\
0 & \lambda_1 & 0\\
0 & 0 & \lambda_2
\end{bmatrix} - \alpha_4 \cdot \begin{bmatrix}
\barn^{3/2} & \barn^{3/2} & \barn\\
\barn^{3/2} & \barn^{3/2} & \barn\\
\barn & \barn & \barn
\end{bmatrix}
\end{align*}
which is PSD given the bounds on $\lambda_1, \lambda_2, \lambda_3$ in conditions \pref{eq:cond2} and \pref{eq:cond3}.  To see this, one
shows that all the $2 \times 2$ principal minors are PSD.

On the other hand, for any  $x \in \R^{\binom{n}{2}}$, we can write
\begin{align*}
x^T B_W x & \geq  \lambda_0 \norm{\Pi_0 x}^2 + \frac{\theta}{2}\norm{\Pi_W \Pi_2 x}^2 - 2\norm{Q} \norm{\Pi_W \Pi_2 x} \norm{\Pi_0 x} \\
& + \lambda_1 \norm{\Pi_1x}^2 + \frac{\theta}{2} \norm{\Pi_W \Pi_2
  x}^2 -  2\norm{Q}\norm{\Pi_W \Pi_2x} \norm{\Pi_1x}
\end{align*}
Now we will appeal to the fact that a quadratic $r(p,q) = ap^2 + 2bpq + cq^2 \geq 0$ for all $p, q \in \R$ if $b^2 <
4ac$ and $a > 0$.  Since $\theta \lambda_1, \theta\lambda_0 \geq 16\norm{Q}^2$ by condition \pref{eq:cond2}, it is easily seen
that the above quadratic form is always non-negative, implying that
$B_W \succeq 0$.
\end{proof}
An immediate corollary of the proof of the above lemma is the
following.
\begin{corollary} \label{corr:calcs}
Under the hypothesis of \pref{lem:calcs}, with probability $1-O(n^{-4})$,
    \[
	 H_{22} - K
\succeq
\frac{1}{2} \cdot \E[H_{22}]
- \frac{16\norm{Q}^2}{\lambda_1}\cdot \Pi_2 \Pi_W \Pi_2.
    \]
\end{corollary}
The above corollary is a consequence of the fact that $H_{22} - K =
B_W+ B_{W^{\perp}} + \frac{1}{2} \E[H_{22}]$.

\subsection{The Corrected Witness}

Suppose we have an unconstrained matrix $M$ that we wish to modify as little as possible so as
to ensure $M \succeq 0$.  Given a test vector $w$ so that $w^T M w < 0$, the
natural update to make is to take $M' = M + \beta \cdot ww^T$ for a
suitably chosen $\beta$.  This would suggest creating a new SDP
solution by setting $H_{22}' = H_{22} + \beta \sum_{i \in [n]} a_i a_i^T $.

Unfortunately, the SOS SDP relaxation has certain hard constraints,
namely that the non-clique entries are fixed at zero.  Moreover,
the entry $X_{S_1,S_2}$ must depend only on $S_1 \cup S_2$. Setting the SDP solution matrix to $H_{22} +
\beta \sum_{i \in [n]} a_i a_i^T$ would almost certainly violate both these constraints.
It is thus natural to consider multiplicative updates to the entries of the
matrix which clearly preserve the zero entries of the matrix.

Specifically, the idea would be to consider an update of the form $M' = M+ \beta D_w M D_w$ where
$D_w$ is the diagonal matrix with entries given by the vector $w$.  If
the matrix $M$ has a significantly large eigenvalue along $\ovec$,
i.e., $M \succeq \lambda_0 \cdot \novec\novec^\top + \Delta$, for some matrix $\Delta$ with $\|\Delta\| \ll \lambda_1$, then this multiplicative update has a
similar effect as an additive update,
$$M' \succeq M + \beta \cdot \lambda_0 \cdot w w^T  + \beta D_w \Delta D_w\mcom$$
where the norm of the final ``error'' term $\beta D_w \Delta D_w$ is relatively small.
Recall that, in our setting, the Deshpande Montanari SDP solution matrix $N$ does
have a large eigenvalue along $\novec$.  We now formally describe our SDP solution, first as a matrix according to the intuition given above, and then as a set of pseudomoments.
\begin{definition}[Corrected SDP Witness, matrix view]\label{def:sol}
Let $\hat a_1,\ldots,\hat a_n \in \R^{\binom{n}{\le 2}}$ be defined so that
\[
    \hat a_{i}(A) = \begin{cases}
	0 & |A| < 2\\
	a_{i}(A) & |A| = 2.
    \end{cases}
\]
Define $\hat D_{i} \in \R^{\binom{n}{\le 2}}$ to be the diagonal matrix with $\hat a_i$ on the diagonal.
Define $\hat K$ to be the restriction of $N(G,\underline{\alpha})$ to the non-multilinear entries.
Also let
\[
    N'(G,\underline{\alpha}) = N(G,\underline{\alpha}) + \beta \cdot \sum_{i\in[n]} \hat D_{i} \left(N(G,\underline{\alpha}) - \hat K\right) \hat D_i,
\]
where $\beta = \frac{1}{100 \sqrt{n}\log n}$.
Then our \emph{SDP witness} is the matrix $M'$, defined so that
\begin{align*}
M'(G,\underline\alpha) = \cP\bigg(N'(G,\underline\alpha)\bigg),
\end{align*}
where $\cP$ is the projection that zeros out rows and columns corresponding to pairs $(i,j) \not\in E$.
\end{definition}
\begin{definition}[Corrected SDP Witness, pseudomoments view]\label{def:sol-pE}
    Let $\beta = \frac{1}{100\sqrt{n}\log n}$, and let $\underline{\alpha} \in \R_+^4$ be a set of parameters, to be fixed later.
    For a subset $S \subseteq [n]$, let $G[S]$ be the graph induced on $G$ by $S$.
For any subset of at most $4$ vertices $S \subset [n]$, $|S| \le 4$, we define
\begin{equation*}
    \pE[X_{S}](G,\underline{\alpha}) =
	\begin{cases}
	    c_4(\underline{\alpha})\cdot \binom{\kappa}{4}\big/\binom{n}{4} + \beta \sum_{v \in [n]} (-1)^{\text{\#\{edges from $v$ to $S$ in $G$\}}}
	    & |S| =4 \text{ and } G[S] =\text{clique},\\
	    c_{|S|}(\underline{\alpha}) \cdot \binom{\kappa}{|S|}\big/\binom{n}{|S|} & |S| \le 3 \text{ and } G[S] =\text{clique},\\
	    0 & \text{otherwise}\mcom
	\end{cases}
\end{equation*}
where $c_{|S|}(\underline{\alpha}) $ is some factor chosen for each $|S| \in \{0,\ldots,4\}$ depending on the choice of $\underline{\alpha}$, which we will set later to ensure that the final moments matrix is PSD.
\end{definition}

\begin{proposition}\label{prop:lin-const}
    For $\beta = \frac{1}{100 \sqrt{n} \log n}$, and $\alpha_4 \le \tfrac{1}{2}$, with probability at least $1-
  O(n^{-5})$, the solution $N'(G,\underline{\alpha})$ does not violate any of the linear constraints of the planted clique SDP.
\end{proposition}
\begin{proof}
  First, $M'(S_1,S_2) = M(S_1,S_2)$ whenever $\vert S_1 \cup S_2 \vert < 4$ so these
  entries satisfy the constraints of the SDP.
  If $\card{S_1 \cup S_2} = 4$ then $M'(S_1,S_2)$ is given by,

$$ M'(S_1,S_2) = \alpha_4 \cdot \Ind[S_1 \cup S_2 \text{ is a clique}] \cdot \left( 1+ \beta \sum_{i \in
    [n]} \prod_{j \in S_1 \cup S_2} A_{ij} \right) \mper$$

Notice that $M'(S_1,S_2)$ is non-zero only if $S_1 \cup S_2$ is a clique, and
it depends only on $S_1 \cup S_2$.  Moreover, $\sum_{i \in [n]} \prod_{j \in
  S_1 \cup S_2} A_{ij}$ is a sum over iid mean $0$ random variables and
therefore satisfies,
$$ \Pr\left\{ \Abs{\sum_{i \in [n]} \prod_{j \in S_1 \cup S_2} A_{ij}}
    \leq 100 \sqrt{n} \log{n} \right\} \geq 1- O(n^{-10}) \mper$$
A simple union bound over all subsets $S_1 \cup S_2 \in \binom{n}{4}$ shows
that $M'(S_1,S_2) \in [0,1]$ for all of them with probability at least $1-
O(n^{-5})$.
\end{proof}

It now remains to verify that $N'(G,\underline{\alpha}) \succeq 0$.
We will do this by verifying the Schur complement conditions, as in \cite{DeshpandeM15}.
Analogous to the submatrix $H_{22}$, one can consider the corresponding submatrix
$H'_{22}$ of $N'$.  The expression for $H'$ is as follows:
\[
    H_{22}' \defeq H_{22} + \sum_{i\in[n]} D_i (H_{22} +\tfrac{1}{16} \alpha_2^2J_{\binom{n}{2}} -K) D_i,
\]
Here $D_i$ is the matrix with $a_i$ on the diagonal, and $K$ is the
matrix corresponding to the non-multilinear entries (entries
corresponding to monomials like $x_a^2 x_b x_c$), and $J_{\binom{n}{2}}$ is the all-$1$s
matrix.
The matrices $H_{12}$ and $H_{11}$ are unchanged, and so we must simply verify that $H_{22}' \succeq H_{12}^\top H_{11}^{-1} H_{12}$ and that $H_{22}'\succeq 0$.

This concludes our proof overview. In \pref{sec:mainproof}, we verify the Schur complement conditions and prove our main result, and in \pref{sec:matrix-conc} we give the random matrix concentration results upon which we rely throughout the proof.

\ifnum\full=1

\section{Proof of the Main Result} \label{sec:mainproof}
In this section, we will demonstrate that $H'_{22} \succeq 0$, and that $H'_{22} \succeq H_{12}^\top H_{11}^{-1} H_{12}$.
This will allow us to conclude that our solution matrix is PSD, and therefore is a feasible point for the degree-4 SOS relaxation.

\paragraph{Parameters}
Before we proceed further, it will be convenient to parametrize
the choice of $\alpha_1,\alpha_2,\alpha_3$ and $\alpha_4$.  In particular, it will be useful to
fix,
\begin{align}
\alpha_1 \defeq  \frac{\rho}{n^{1/2}} \qquad  \alpha_2 \defeq \frac{\gamma \rho^2}{n} \qquad
\alpha_3 \defeq     \frac{\gamma^3 \rho^3}{n^{3/2}}  \qquad \alpha_4 \defeq \frac{\gamma^6
\rho^4}{n^2} \mper \label{eq:alpha-choice}
\end{align}
for two parameters $\gamma,\rho$, which we will finally fix to $\gamma = \log^{4} n$ and $\rho = \log^{-20}n$.  For
this setting of parameters,  the eigenvalues $\lambda_0,\lambda_1,\lambda_2$ from \pref{prop:eigenspaces} are bounded by,
\begin{align}
\lambda_0 \geq \frac{\alpha_4 n^2}{64} = \frac{\gamma^6 \rho^4}{64} \qquad
\lambda_1 \geq \frac{\alpha_3 n}{4} = \frac{\gamma^3 \rho^3}{4 n^{1/2}} \qquad
\lambda_2 \geq \frac{\alpha_2}{2} =  \frac{\gamma \rho^2}{2n} \label{eq:lambdas}
\end{align}
When convenient, we will also use the shorthand $c_1 \defeq n^{1/2} \alpha_1$, $c_2 \defeq n \alpha_2$, $c_3 \defeq n^{3/2}\alpha_3$, and $c_4 \defeq n^2\alpha_4$.

\subsection{Proving that $H'_{22} \succeq 0$}

Here we will make a first step towards verifying the Schur complement conditions of \pref{lem:schur} by showing that $H'_{22} \succeq 0$.
Specifically, we will show the following stronger claim.

\begin{theorem}\label{thm:psd}
For $\beta = \frac{1}{100 \sqrt{n} \log{n}}$ and $\gamma = \log^{4} {n}$, $\rho <
\log^{-20}n$, the following holds with probability at least $1- O(n^{-4})$,

$$ H'_{22} \succeq \frac{1}{8} \E[H_{22}] + \frac{\beta  \lambda_0}{16} \cdot \Pi_W $$

\end{theorem}

\begin{proof}
Fix $\theta = \frac{\norm{Q}^2}{\lambda_1}$.
By definition of $H'_{22}$, we have
    \begin{align*}
	H'_{22}
	&=  H_{22}  + \beta \cdot \sum_{i\in[n]}  D_i (H_{22} +\tfrac{1}{16} \alpha_2^2 J- K)D_i .
    \end{align*}
   Define $P_W = \sum_{i\in[n]} a_i a_i^T$.
    We can apply \pref{lem:calcs} to the $H_{22}$ term and
    \pref{corr:calcs} for $H_{22} - K$,
    \begin{align*}
	H'_{22}
	&\succeq  \frac{1}{4} \E[H_{22}] - \theta\cdot \Pi_2 \Pi_W
          \Pi_2  + \beta \cdot \sum_{i\in[n]}  D_i \left(\frac{1}{2}\E[H_{22}] -\theta
          \Pi_2 \Pi_W \Pi_2 + \frac{1}{16} \alpha_2^2 J\right)D_i .\\
	&\succeq  \frac{1}{4} \E[H_{22}] - \theta \cdot \Pi_2 \Pi_W
          \Pi_2  + \beta \cdot \sum_{i\in[n]}  D_i \left(\frac{\lambda_0}{2} \Pi_0 -\theta
          \Pi_2 \Pi_W \Pi_2 \right)D_i . \qquad(\text{dropping } \Pi_1,\Pi_2, J)\\
      & \succeq    \frac{1}{4} \E[H_{22}]  -  \theta \cdot \Pi_2 \Pi_W
          \Pi_2 + \beta \frac{\lambda_0}{4n^2} P_W  - \beta \theta \sum_{i \in [n]} D_i \Pi_2 \Pi_W \Pi_2 D_i
        \qquad (\text{using} D_i \Pi_0 D_i = a_i a_i^T / \binom{n}{2} )
    \end{align*}
Now we will appeal to a few matrix concentration bounds that we show in \pref{sec:matrix-conc}.
First, with probability $1- O(n^{-5})$, the vectors $\{a_{i}^{\otimes
  2}\}$ are nearly orthogonal, and therefore form a well-conditioned basis for the subspace $W$.
$$ P_W \succeq \frac{n^2}{2}  \cdot \Pi_W.  \qquad (\text{see \pref{lem:w-proj}}) \mper $$
Also, the vectors $\{a_i^{\otimes 2}\}$ have negligible projection on to the
eigenspaces $V_0, V_1$ which implies that with overwhelming probability,
$$\left\|\Pi_2 \Pi_W \Pi_2 - \Pi_W \right\| \le \frac{\log^2 n}{n}    \cdot \Id \mcom \qquad (\text{see \pref{lem:01w}}) \mper$$
Finally, $W$ is an $n$ dimensional space.  Each $D_i \Pi_2 \Pi_W
\Pi_2 D_i$ has only $n$ non-zero singular values each of which is $O(1)$.
Moreover, multiplying on the left and right by $D_i$ acts as a random
linear transformation/ random change of basis.  Intuitively, this
suggests that $\sum_{i }D_i \Pi_2 \Pi_W \Pi_2 D_i$ has $n^2$ eigenvalues
all of which are roughly $O(1)$.  In fact,
with probability $1-O(n^{-5})$,
$$ \sum_i D_i \Pi_2 \Pi_W \Pi_2 D_i \preceq O(n) \cdot \Pi_0 + \log^2{n} \cdot \Id \qquad (\text{see \pref{lem:W2}})$$
Substituting these bounds we get,
\begin{align*}
H' \succeq \frac{1}{4} \E[H_{22}] + \left( \frac{\beta  \lambda_0}{8} -
  \theta \right) \cdot \Pi_W - \left(  \frac{\theta \log^2 n}{n} + \beta \theta \log^2 n
  \right) \cdot \Id - \beta \theta \cdot O(n) \cdot \Pi_0
\end{align*}
By \pref{lem:qkbounds}, with probability at least $1-O(n^{-4})$,
$\norm{Q} \lesssim \alpha_4 \barn^{3/2}$.  Substituting this bound
for $\theta = \frac{\norm{Q}^2}{\lambda_1}$ along with \eqref{eq:lambdas}, finishes
the proof for our choice of parameters.
The details are presented below for completeness.
$$\frac{\theta}{\beta \lambda_0}  \lesssim \frac{\alpha_4^2 \barn^3}{\alpha_3 n} \cdot
\frac{\sqrt{n} \log n}{ \alpha_4 n^2} = {\log^4 n} \cdot \gamma^3\rho
 \ll 1 \mcom$$
$$ \frac{\beta \theta n }{\lambda_0} \lesssim \frac{1}{\sqrt{n} \log n} \cdot \frac{\alpha_4^2 \barn^3}{
  \alpha_3 n } \cdot n \cdot \frac{1}{\alpha_4n^2} \lesssim \log^2{n}\cdot \gamma^3 \rho \ll 1$$
Clearly $\lambda_0 > \lambda_1 > \lambda_2$ and
$$  \lambda_2 \gg 100 \theta \left(\beta \log^2{n}+ \frac{\log^2 n}{n}\right) \mcom$$
because
$$ \frac{\theta \beta \log^2{n}}{\lambda_2} \lesssim  \frac{\alpha_4^2 \barn^3}{\alpha_3 n} \cdot \frac{1}{\sqrt{n} \log n} \cdot
\frac{\log^2{n}}{\alpha_2}   = {\log^4{n}} \cdot \gamma^8 \rho^3 \ll 1 \mper $$

\end{proof}

\subsection{Bounding singular values of $H_{12}$}
Towards bounding the eigenvalues of $H_{21}^T H_{11}^{-1} H_{12}$,
Deshpande and Montanari \cite{DeshpandeM15} observe the following
properties of $H_{21}$ with regards to the spaces $V_0,V_1$
and $V_2$.
\begin{lemma}[Consequence of Propositions 4.18 and 4.27 in \cite{DeshpandeM15}]\label{lem:H12-stuff}
    Let $Q_n\in\R^{n\times n}$ be the orthogonal projector to the space spanned by $\ovec_n$.
    Let $p = \tfrac{1}{2}$.
    For the matrix $H_{21}$, we have that for sufficiently large $n$, with probability $1-O(n^{-5})$,
    $$
	\|\E[H_{21}] - H_{21}\| \le \alpha_3 \barn\mcom$$ and
    \begin{align*}
	\|\Pi_0 \E[H_{21}] Q_n\| &\le \tfrac{1}{4} n^{3/2}\alpha_3 + 2\alpha_2 n^{1/2}
	&	\Pi_0 \E[H_{21}] Q_n^{\perp} &= 0\\
	\Pi_1 \E[H_{21}] Q_n &= 0 &
	\|\Pi_1 \E[H_{21}] Q_n^{\perp}\| &\le \alpha_2 n^{1/2} \\
	\Pi_2 \E[H_{21}] Q_n &= 0 &
	\Pi_2 \E[H_{21}] Q_n^\perp &= 0.
\end{align*}
\end{lemma}
Unfortunately, the bound of \cite{DeshpandeM15} on $\norm{\E[H_{21}]- H_{21}}$ is insufficient for our purposes, and we
require a more fine-grained bound on the deviation from expectation.  In fact,
outside the problematic subspace $W$, we show that $H_{21} -
\E[H_{21}]$ is much better behaved.

\begin{proposition}\label{prop:H12P2}
    Let $W = \Span_{i\in[n]}\{a_i\}$, and let $\Pi_W$ be the projector to that space.
    With probability at least $1-O(n^{-4})$, the following holds for every $x \in
    \R^{\binom{n}{2}} $
   \[
    	\|x^T \Pi_2 (H_{21} - \E[H_{21}])\|^2 \lesssim \alpha_3^2 \barn^2 \left(
 \|\Pi_W x\|^2  +
\frac{\log^2{n}}{n} \norm{x}^2 \right)
    \]
\end{proposition}

\begin{proof}
    From \pref{lem:H12-stuff},
    we have that
    \[
	\Pi_2 \E[H_{21}] = 0.
    \]
    Thus we may work exclusively with the difference from the expectation;
    for convenience, we let $U = H_{21} - \E[H_{21}]$.
    Fix $A = \{a,b\}$ and $c$, so that $|\{a,b\}| = 2$.
    By inspection, the entry $(A,c)$ of $U$ is given by the polynomial
    \[
	U_{A,c} = \frac{\alpha_3}{4}(A_{ac}A_{bc} + (A_{ac} + A_{bc})).
    \]
    Thus, the columns of $U$ are in $W \cup V_1$.
    So we have that
    \begin{align*}
	\Pi_2 U
	&= \Pi_2(\Pi_W + \Pi_{V_1\setminus W} + \Pi_{W\cup V_1}^{\perp}) U \\
	&= \Pi_2 \Pi_W U  + \Pi_2\Pi_{V_1\setminus W} U + \Pi_{W\cup V_1}^{\perp}U\\
	&= \Pi_2\Pi_W U ,
    \end{align*}
    where the latter two terms were eliminated because $V_1 \perp V_2$ and $W\cup V_1 \perp \Span\{\col(U)\}$.

    In \pref{lem:01w},
    we bound $\|\Pi_{01} \Pi_W \Pi_{01}\| \le O(\frac{\log^2 n}{ n})$.
    So we have that
    \begin{align*}
	\|x^T \Pi_2 \Pi_W\|^2
	&\le x^T(\Id -\Pi_{01})\Pi_W(\Id- \Pi_{01})x\\
	&= \|\Pi_W x\|^2 - 2x^T\Pi_{01}\Pi_W x + x^T\Pi_{01} \Pi_W \Pi_{01} x\\
	&\le \|\Pi_W x\|^2 + 2\|\Pi_W x\|\|\Pi_{01}\Pi_W x\| + \norm{x}^2 \|\Pi_{01}\Pi_W\Pi_{01}\|\\
	&\lesssim
 2\|\Pi_W x\|^2 + \frac{2\log^2n}{n} \cdot \norm{x}^2,
    \end{align*}
    where we have applied the inequality $a^2 + b^2 \ge 2ab$.
    The conclusion follows by noting that $\|x^T\Pi_2 H_{21}\|  =
    \|x^T\Pi_2\Pi_W U\| \le \|x^T\Pi_2\Pi_W\|\cdot \|U\|$, and that by \pref{lem:H12-stuff}
    $\|U\| \lesssim \alpha_3 \barn $.
\end{proof}

\subsection{ Bounding singluar values of $H_{12}^T H_{11}^{-1} H_{12}$}

   We will bound $H_{21}H_{11}^{-1} H_{12}$, as in our bounds on $H_{22}'$, by splitting the matrix up according to the eigenspaces $\Pi_0,\Pi_1,\Pi_2$.

   \begin{theorem}\label{thm:h21h11h12}
       Let $c_1,\ldots, c_4$ be as defined in \pref{eq:alpha-choice}. For the choice of $\underline{\alpha}$ in \pref{eq:alpha-choice}, we have that with probability $1-O(n^{-5})$,
\begin{align*}
H_{12}^T H_{11}^{-1} H_{12} \preceq \frac{c_3^2}{c_2} \cdot \Pi_0
  + \frac{c_2^2 + c_3^2 \log^2 n}{c_1 n^{1/2}} \cdot \Pi_1 + \frac{c_3 \log^{4}
  n}{ c_1 n^{3/2}} \Pi_2 + \frac{c_3^2 \log^2 n}{c_1 n^{1/2}} \cdot \Pi_W
\end{align*}
   \end{theorem}
   \begin{proof}
    For each $a,b\in\{0,12\}$, define the matrix $U_{ab} = \Pi_a
    (\E[H_{21}] - H_{21}) \Pi_b$.
    We can verify that for our choice of $\underline\alpha$ the conditions of \pref{lem:H11} are met, and so we conclude that with probability
    $1-O(n^{-5})$,
    $$ H_{11}^{-1} \preceq \frac{1}{n(\alpha_2 - 2\alpha_1^2)}  \cdot Q_n
    + \frac{1}{\alpha_1} \cdot Q_n^{\perp}$$

    For $x \in \R^{\binom{n}{2}}$, let $y_{U} = (\E[H_{12}] - H_{12})
    \Pi_2x$, $z_U = (\E[H_{12}] - H_{12}) (\Pi_0 + \Pi_1)x$,
    $x_{A} = \E[H_{12}] \Pi_0 x$, $x_{B} =  \E[H_{12}] \Pi_1 x$ and $x_{C} =
     \E[H_{12}] \Pi_2 x$.
In order to simplify our computations, we will use the following
observation.
\begin{observation} \label{obs:simple-cauchy-schwartz}
  Given $A \in \R^{m \times m}$ with $A \succeq 0$ and vectors
  $x_1,\ldots, x_t \in \R^m$,
  $$ \left( \sum_{i \in [t]} x_i \right)^T A \left( \sum_{i \in [t]} x_i \right)
  \leq t \cdot \sum_{i \in [t]} x_{i}^T A x_i  $$
\end{observation}
\begin{proof}  If we set $y_i = A^{1/2}x_i$ then the inequality reduces
  to,
$$ \norm{ \sum_{i \in [t]} y_i}^2 \leq t \sum_{i \in [t]}
\norm{y_i}^2 \mcom$$
which is an immediate consequence of triangle inequality for $\norm{\cdot}$
and Cauchy-Schwartz inequality.
\end{proof}

Using \pref{obs:simple-cauchy-schwartz},
    $$x^T H_{12}^T H_{11}^{-1} H_{12} x \leq 5 \left(   x_{A}^T H_{11}^{-1} x_{A} + x_B^T H_{11}^{-1} x_B + x_{C}^T
    H_{11}^{-1} x_C  + y_U^T H_{11}^{-1} y_U + z_{U}^T H_{11}^{-1} z_U\right)\mper$$

  To simplify the calculations and make the dominant terms apparent,
  let us fix $ \alpha_1 = c_1 / \sqrt{n}$, $\alpha_2 = c_2/n$, $\alpha_3 = c_3/n^{3/2}$ and
  $\alpha_{4} = c_4/n$ wherein each $c_i \in [ 1/\log^{200} n,1 ]$.
  First, observe that $\frac{1}{\alpha_1} \gg \frac{1}{n(\alpha_2 -
    2\alpha_1^2)}$ for this setting of parameters.

  For the terms $x_A^T H_{11}^{-1} x_A$, $x_B^T H_{11}^{-1} x_B$ and
  $x_C^T H_{11}^{-1} x_C$ we can write,
  \begin{align*}
    x_{A}^T H_{11}^{-1} x_A & \lesssim \norm{\Pi_0x}^2 \left(
    \frac{1}{n(\alpha_2 - 2\alpha_1^2)} \cdot \norm{\Pi_0
    \E[H_{21}] Q_n}^2 +  \frac{1}{\alpha_1}  \cdot \norm{\Pi_0
    \E[H_{21}] Q^{\perp}_n}^2 \right) \nonumber \\
    & \lesssim \frac{(\alpha_2  + \alpha_3 n)^2}{(\alpha_2 -2 \alpha_1^2)} \cdot
      \norm{\Pi_0 x}^2 \lesssim \frac{c_3^2}{c_2} \norm{\Pi_0x}^2 \\
    x_{B}^T H_{11}^{-1} x_B  & \lesssim \norm{\Pi_1x}^2 \left(
    \frac{1}{n(\alpha_2 - 2\alpha_1^2)} \cdot \norm{\Pi_1
    \E[H_{21}] Q_n}^2 +  \frac{1}{\alpha_1}  \cdot \norm{\Pi_1
    \E[H_{21}] Q^{\perp}_n}^2 \right) \nonumber \\
  & \lesssim \frac{\alpha_2^2 n}{\alpha_1} \norm{\Pi_1x}^2 \lesssim
    \frac{c_2^2}{c_1 n^{1/2}} \cdot \norm{\Pi_1 x}^2 \\
     x_{C}^T H_{11}^{-1} x_C & \lesssim \norm{\Pi_2x}^2 \left(
    \frac{1}{n(\alpha_2 -2\alpha_1^2)} \cdot \norm{\Pi_2
    \E[H_{21}] Q_n}^2 +  \frac{1}{\alpha_1}  \cdot \norm{\Pi_2
    \E[H_{21}] Q^{\perp}_n}^2 \right) \nonumber \\
    & = 0 \\
     z_{U}^T H_{11}^{-1} z_U & \lesssim \norm{ (\Pi_0+\Pi_1)x}^2 \left(
    \frac{1}{n(\alpha_2 -2\alpha_1^2)} \cdot \norm{
    H_{21} - \E[H_{21}]}^2 \norm{Q_n}^2 +  \frac{1}{\alpha_1}  \cdot
                              \norm{H_{21} - \E[H_{21}]}^2
                              \norm{Q^{\perp}_n}^2 \right) \nonumber \\
    & \lesssim \left( \norm{\Pi_0x}^2 + \norm{\Pi_1x}^2\right) \left(
      \frac{\alpha_3^2 \barn^2}{n(\alpha_2 - 2\alpha_1^2)} + \frac{\alpha_3^2 \barn^2}{\alpha_1}  \right)
      \lesssim   \left( \norm{\Pi_0x}^2 + \norm{\Pi_1x}^2\right)  \cdot \frac{c_3^2 \log^2 n}{c_1
      n^{1/2}} \\
\intertext{Finally, we have}
y_U^T H_{11}^{-1} y_{U} & \lesssim \norm{(H_{12} - \E[H_{12}]) \Pi_2
                         x}^2 \cdot \left( \frac{1}{n(\alpha_2 - 2\alpha_1^2)} + \frac{1}{\alpha_1}\right) \nonumber \\
  & \lesssim  \left( \frac{\alpha_3^2 \barn^2}{n(\alpha_2 - 2\alpha_1^2)} +
    \frac{\alpha_3^2 \barn^2}{\alpha_1}\right)  \cdot \norm{\Pi_W x}^2
    +  \left( \frac{\alpha_3^2 \barn^2 \log^2 n}{n^2(\alpha_2 - 2\alpha_1^2)} + \frac{
    \alpha_3^2 \barn^2 \log^2
    n}{n \alpha_1}\right) \cdot \norm{x}^2 \nonumber \\
& \lesssim \frac{c_3^2 \log^2n}{c_1 n^{1/2}} \cdot \norm{\Pi_W x}^2 +
  \frac{c_3^2 \log^4 n}{c_1 n^{3/2}} \cdot \norm{x}^2
\end{align*}
The conclusion follows by grouping the projections, and taking the dominating terms as $n$ grows.
\end{proof}

\subsection{Proof of $H'_{22} \succeq H_{12}^T H_{11}^{-1} H_{12}$}
\begin{theorem}\label{thm:schur-2}
    For the choice of $\underline\alpha$ given in \pref{eq:alpha-choice}, we have that $H'_{22} \succeq H_{12}^T H_{11}^{-1} H_{12}$ with probability $1-O(n^{-4})$.
\end{theorem}
\begin{proof}
Recall that by \pref{thm:psd} with probability at least $1 - O(n^{-4})$,
\begin{align*}
 H_{22}' \succeq \frac{\lambda_0}{4} \cdot \Pi_0 + \frac{\lambda_1}{4}
  \cdot \Pi_1 + \frac{\lambda_2}{4} \cdot \Pi_2 + \frac{\beta \lambda_0}{16} \Pi_W \mper
\end{align*}
By our choice of the parameters $\alpha_1,\alpha_2,\alpha_3, \alpha_4$,
the conclusion of \pref{thm:h21h11h12} implies that
\begin{align*}
H_{22}' - H_{12}^{T} H_{11}^{-1} H_{12}
&\succeq
\left(\frac{1}{4}\lambda_0 - \frac{c_3^2}{c_2}\right)\Pi_0 +
\left(\frac{1}{4}\lambda_1 - \frac{c_2^2 + c_3^2\log^2 n}{c_1n^{1/2}}\right)\Pi_1\\
&\qquad +
\left(\frac{1}{4}\lambda_2 - \frac{c_3\log^4 n}{c_1n^{3/2}}\right)\Pi_2 +
\left(\frac{1}{16}\beta\lambda_0 - \frac{c_3^2\log^2 n}{c_1 n^{1/2}}\right)\Pi_W\\
&\succeq 0\mcom
\end{align*}
as desired.
The details of the calculation are spelled out below for the sake of completeness; we verify that the coefficient of each projector is non-negative.

\noindent For the space $\Pi_0$,
\begin{align*}
\frac{c_3^2}{c_2} \cdot \frac{1}{\lambda_0} = \frac{c_3^2}{c_2 c_4} =
  \gamma^{-1} \ll 1
\end{align*}
For $\Pi_1$,
\begin{align*}
 \frac{c_2^2 + c_3^2 \log^2 n}{c_1 n^{1/2}}  \cdot \frac{1}{\lambda_1} \lesssim \frac{c_2^2 + c_3^2
  \log^2 n}{c_3 c_1 } = \frac{1}{\gamma} + \gamma^3 \rho^2 \log^2n \ll 1 \mper
\end{align*}
For $\Pi_2$,
\begin{align*}
 \frac{c_3 \log^{4}
  n}{ c_1 n^{3/2}} \cdot \frac{1}{\lambda_2} \leq \frac{1}{n^{1/2}} \frac{c_3 \log^4
  n}{c_1 c_2} \ll 1 \mper
\end{align*}
Finally for $\Pi_W$,
\begin{align*}
 \frac{c_3^2 \log^2 n}{c_1 n^{1/2}} \cdot \frac{1}{\beta \lambda_0} = \frac{c_3^2\log^{3}n}{c_1
  c_4} = \rho \log^{3}n \ll 1 \mper
\end{align*}
This concludes the proof.
\end{proof}

\subsection{Proof of Main Theorem}
We finally have the components needed to prove \pref{thm:main-technical}.

\begin{proof}[Proof of \pref{thm:main-technical}]
    First, we recall that independent of our choice of $\underline \alpha$, the SDP solution defined in \pref{def:sol} does not violate any of the linear constraints of \pref{eq:sdprelaxation}, as shown in \pref{prop:lin-const}.
    To meet the program constraints, it remains to show that for the choice of $\underline \alpha$ given in \pref{eq:alpha-choice}, our solution is PSD.

    The solution matrix $M'(G,\underline\alpha)$ is a principal submatrix of $N'(G,\underline \alpha)$, and so $N' \succeq 0$ implies $M' \succeq 0$.
    We prove that $N'$ satisfies the Schur complement conditions from \pref{lem:schur} with high probability.
    Observing that $H'_{11} = H_{11}$ and $H_{12}' = H_{12}$, we apply \pref{thm:schur-2}, which states that for our choice of $\underline \alpha$, $H_{22}' \succeq H_{12}^T H_{11}^{-1} H_{21}$.
    For the $H_{11}$ term, we apply the lower bound on the eigenvalues of $H_{11}$ given by \cite{DeshpandeM15} (\pref{lem:H11}), which state that so long as $\alpha_1 - \alpha_2 \gg \alpha_2 n^{-1/2} $ and $\alpha_2 - 2\alpha_1^2 \ge 0$, we have $H_{11} \succeq 0$ with probability $1-O(n^{-5})$.
    For our choice of $\alpha_1,\alpha_2$, we have
    \[
	\frac{\alpha_2}{ n^{1/2}} = \frac{\gamma \rho^2}{ n^{3/2}} \ll \frac{\rho}{ n^{1/2}} - \frac{\gamma \rho^2}{ n} = \alpha_1 - \alpha_2\mcom
    \]
    and
    \[
	\alpha_2 - 2\alpha_1^2 = \frac{\gamma \rho^2 - 2\rho^2}{n} \gg 0,
    \]
    and so we may conclude that $H_{11} \succeq 0$.
    Therefore by a union bound, the conditions of \pref{lem:schur} are satisfied with probability $1-O(n^{-4})$, and $N'(G,\underline\alpha) \succeq 0$  and so our solution satisfies the PSDness constraint.

    It remains only to prove that the objective value is $\tilde\Omega(\sqrt{n})$.
    The objective value is simply $\sum_{i\in[n]} \alpha_1 = n \alpha_1 = \rho n^{1/2}$, concluding our proof.
\end{proof}

\section{Concentration of Projected Matrices}\label{sec:matrix-conc}
In this section, we give bounds on the spectra of random matrices that
are part of the correction term.
Though we are able to recycle many of the spectral bounds of Deshpande and Montanari \cite{DeshpandeM15}, in our modification to their witness, we introduce new matrices which also require description and norm bounds.

We obtain our bounds by employing the trace power method.
The trace power method uses the fact that if $X \in \R^{n\times n}$ is a symmetric matrix, then for even powers $k$, $\Tr(X^k) = \sum_{i \in [n]} \lambda_i(X)^k \ge \lambda_{\max}(X)^k$.
By bounding $\E[\Tr(X^k)]^{1/k}$ for a sufficiently large $k$, we
essentially obtain bounds on the infinity norm of the vector of
eigenvalues, i.e., a bound on the spectral norm of the matrix $X$.
A formal statement follows, and the proof is given in \pref{app:misc} for completeness.

\begin{lemma}
    \label{lem:trace-power}
    Suppose an $n \times n$ random matrix $M$ satisfies $\E[\Tr(M^k)] \le n^{\alpha k + \beta}\cdot (\gamma k)!$ for any even integer $k$, where $\alpha, \beta,\gamma$ are constants.
    Then $$ \Pr\Paren{\|M\| \lesssim \eta^{-1/\log n} \cdot n^{\alpha} \cdot \log n} \ge 1-\eta \mper $$
\end{lemma}

Our concentration proofs consist of, for each matrix in question,
obtaining a bound on $\E[\Tr(X^k)]$.
The expression $\E[\Tr(X^k)]$ is a sum over products along closed paths of length $k$ in the entries of $X$.
In our case, the entries of the random matrix $X$ are themselves
low degree polynomials in random variables $\{ A_{ij}\}_{i \in [n], j
  \in [n]}$ where $A_{ij}$ is the centered random variable that
indicates whether the edge $(i,j)$ is part of the random graph $G$.
Thus $\Tr(X^k)$ can be written out as a polynomial in the random
variables $\{A_{ij}\}_{i,j \in [n]}$.  Since the random variables $\{A_{ij}\}_{i,j \in
  [n]}$ are centered (i.e., $\E[A_{ij}] = 0$), almost all of the terms
in $\E[\Tr(X^k)]$ vanish to zero.  The nonzero terms are precisely those monomials in which every variable appears with even multiplicity.

For the purpose of moment calculations, we borrow much of our
terminology from the work of Deshpande and Montanari \cite{DeshpandeM15}.
Every monomial in random variables $\{A_{ij}\}_{i, j \in [n]}$
corresponds to a {\it labelled graph} $(F = (V,E), \ell)$ that
consists of a graph $F=(V,E)$ and a labelling
$\ell:V \to [n]$ that maps its vertices to $[n]$.  A labelling of $F$
{\it contributes} (is nonzero in expectation), if and only if every pair $\{i,j\}$
appears an even number of times as a label of an edge in $F$.
The problem of bounding $\E[\Tr(X^k)]$ reduces to counting the number
of the number of contributing labeled graphs.

For example, for a given matrix $X$, we may have a bound on the number
of ``vertices'' and ``edges'' in a term of $\E[\Tr(X^k)]$ as a function of $k$.
In that case, we may use the following proposition, which allows us to bound the number of such graphs in which every variable $A_{ij}$, corresponding to an edge between vertices $i$ and $j$, appears at least twice.

\begin{proposition}
    \label{prop:eulerian}

Let $F = (V,E)$ be a multigraph and let $\ell: V\to [n]$ be a labelling
such that each pair $(i,j)$ appears an even number of times as the
label of an edge in $E$. Then,
$$ \card{\{\ell(v) | v \in V\} } \leq \frac{|E|}{2} + (\# \text{ connected components
of F})$$
\end{proposition}

\begin{proof}
From $F$, we form a new graph $F'$ by identifying all the nodes with
the same label; thus, the number of nodes in $F'$ is the number of
labels in $F$.  We then collapse the parallel edges in $F'$ to form the
graph $H$; since each labelled edge appears at least twice, the number
of edges in $H$ is at most half that in $F$.
The number of nodes in $H$ (and thus labels in $F$) is at most the
number of edges in $H$ plus the number of connected components; this
is tight when $H$ is a forest.  Thus the number of distinct labels in
$F$ is at most $|E|/2 + c$, where $c$ is the number of components in $F$.
\end{proof}

We apply this lemma, as well as simple inductive arguments, to bound the number of contributing terms in $\E[X^k]$ for the matrices $X$ in question, and this allows us to bound their norms.
We give the concentration proofs the following subsection.

\subsection{Proofs of Concentration}
Let $G$ be an instance of $\mathbb{G}(n,\tfrac{1}{2})$.
As in the preceeding sections, define the vector $A_i \in \R^n$ so that $A_i(j) = 1$ if $(i,j) \in E(G)$, $A_i(j) = -1$ if $(i,j) \not \in E(G)$, and $A_i(i) = 0$.
Again as in the preceeding sections, define $a_1,\ldots,a_n\in\R^{\binom{n}{2}}$
by setting $a_i$ to be the restriction of
$A_i^{\tensor 2}$ to coordinates corresponding to unordered pairs,
i.e., $a_{i}(\{c,d\}) = A_{ic}A_{id}$ for all $\{c,d\} \in \binom{n}{2}$.
We will continue to use the notation $W = \Span_{i\in[n]}(a_i)$, and the notation $D_i = \diag(a_i)$.

We begin with a lemma that shows that the $a_i$ are close to an orthogonal basis for $W$.
\begin{lemma}
    \label{lem:w-proj}
    If $P_W \defeq \sum_{i} a_i a_i^T$ then with probability at least $1 - O(n^{-5})$,
    \[
	\tfrac{1}{n^2}(1+o(1))\cdot P_W \succeq \Pi_W \succeq (1-o(1))\tfrac{1}{n^2}\cdot P_W
    \]
\end{lemma}
\begin{proof}
    By definition, the vectors $a_1,\ldots,a_n$ form a basis for the
    subspace $W$.

    Let $\mathcal{R}$ be the matrix whose $i$th row is $a_i$.
    We will use matrix concentration to analyze the eigenvalues of $\cR \cR^T$, which are identical to the nonzero eigenvalues of $P_W = \cR^T \cR$.

    The $(i,j)$th entry of $\cR \cR^T$ is $\iprod{a_i,a_j} = \frac{1}{2} \iprod{A_i^{\tensor 2}, A_j^{\tensor 2}} = \frac{1}{2} \iprod{A_i,A_j}^2$.
    When $i = j$, this is precisely $\tfrac{1}{2}(n-1)^2$, and so $2\cR \cR^T = (n-1)^2 \cdot \Id_n + B$, where $B$ is a matrix that is $0$ on the diagonal and equal to $\iprod{A_i^{\tensor 2}, A_j^{\tensor 2}}$ in the $(i,j)$th entry for $i \neq j$.

    Let $M = B - \E[B] = B - (n-2)(J_n-\Id_n)$.
    We will use the trace power method to prove that $\| M \| = O(n^{3/2})$.
    The $(i,j)$th entry of $M$ is given by $0$ for $i=j$, and when $i\neq j$
\[
	M(i,j) = \iprod{A_i, A_j}^2 - (n-2) = \Paren{\sum_{p,q} A_{ip}A_{iq}
	A_{jp}A_{jq}} - (n-2) = \sum_{p\neq q} A_{ip}A_{iq}A_{jp}A_{jq}.
    \]

 The expression $\Tr(M^k)$ is a sum over monomial products over variables
    $\{A_{ip}\}_{i,p\in[n]}$, where each monomial product corresponds to a
    labelling $\cL:F \to [n]$ of a graph $F$.
    Each entry in $M_{ij}$ corresponds to a sum over {\it links}, where each link is a cycle of length $4$, with the vertices $i,j$ on opposite ends of the cycle, and the necessarily distinct vertices $p,q$ are on the other opposite ends of a cycle.
  We will refer to $i,j$ as the {\it center vertices} and $p,q$ as
    the {\it peripheral} vertices of the link.
Each edge $(u,v)$ of the link is weighted by $A_{uv}$.
    Since $A_{ii} = 0$ for all $i \in [n]$, for every {\it
      contributing labelling}, it can never be the case that one of $p,q = i$.
       Each monomial product in the summation $\Tr(M^k)$ corresponds to a labelling $(F,\cL)$ of the graph $F$,
    where $F$ is a cycle with $k$ links.
    $F$ has $4k$ edges, and in total it has $3k$ vertices.
    
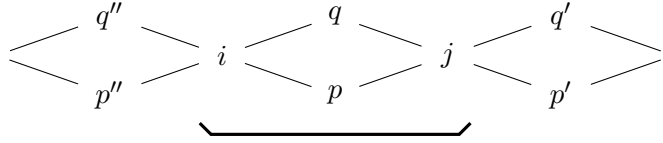
\begin{center}
    
\begin{center}
\begin{tikzpicture}[scale = 0.75]
    \def\hspc{2}
    \def\vspc{0.7}
    \def\shift{4}
\draw (\hspc, -\vspc) node[circle,fill=white] (d) {$p$}
--(0,0) node[circle,fill=white] (i) {$i$}
-- (\hspc, \vspc) node[circle,fill=white] (c) {$q$};
\draw
(\shift+ \hspc, -\vspc) node[circle,fill=white] (d2) {$p'$}
--(\shift,0) node[circle,fill=white] (i2) {$j$}
-- (\shift + \hspc, \vspc) node[circle,fill=white] (c2) {$q'$};
\draw (-\shift+ \hspc, -\vspc) node[circle,fill=white] (d0) {$p''$}
--(-\shift,0) node[circle,fill=white] (i0) {};
\draw (i0) -- (-\shift + \hspc, \vspc) node[circle,fill=white] (c0) {$q''$};
\draw (c0)--(i);
\draw (d0)--(i);
\draw (c)--(i2);
\draw (d)--(i2);
\draw[very thick] (-.2*\hspc,- \vspc - 0.5)--(-.1*\hspc ,- \vspc - 0.7) -- (1.1*\hspc + \hspc,-\vspc-0.7) -- (1.2*\hspc + \hspc ,- \vspc - 0.5);

\draw
(c2)--
(\shift +2*\hspc,0) node[circle,fill=white] (c3) {}
-- (d2);
\end{tikzpicture}
\end{center}

\end{center}

    The quantity $\Tr(M^k)$ is equal to the sum over all labellings of $F$.
    Taking the expectation, terms in $\E[\Tr(M^k)]$ which contain a variable $A_{uv}$ with multiplicity 1 have expectation $0$.
    Thus, $\E[\Tr(M^k)]$ is equal to the number of labellings of $F$ in which every edge appears an even number of times.

    We prove that any such contributing labelling $\cL:F \to [n]$ has at most $3k/2 + 1$ unique vertex labels.
    We proceed by induction on $k$, the length of the cycle.
    In the base case, we have a cycle on two links; by inspection no such cycle can have more than $5$ labels, and the base case holds.

    Now, consider a cycle of length $k$.
    If every label appears twice, then we are done, since there are
    $3k$ vertices in $F$.
    Thus there must be a vertex that appears only once.

    There can be no {\it peripheral} vertex whose label does not repeat, since the two center vertices neighboring a single peripheral vertex cannot have the same label in a contributing term, as $M(i,i) = 0$.
    Now, if there exists a {\it center vertex} $i$ whose label does not repeat, it must be that there is a matching between its $p,q$ neighbors so that every vertex is matched to a vertex of the same label; we identify these same-label vertices and remove $i$ and two of its neighbors from the graph, leaving us with a cycle of length $k-1$, having removed at most one label from the graph.
    The induction hypothesis now applies, and we have a total of at most $3(k-1)/2 + 2 \le 3k/2 +1$ labels, as desired.

    Thus, there are at most $3k/2 + 1$ unique labels in any contributing term of $\E[\Tr(M^k)]$.
    We may thus conclude that $\E[\Tr(M^k)] \le n^{3k/2 + 1} \cdot (3k/2 + 1)^{3k}$, and applying \pref{lem:trace-power}, we have that $\|M\| \lesssim \barn^{3/2}$ with probability at least $1 - O(n^{-5})$.

    Therefore, $2 \cR \cR^T = ((n-1)^2 - n + 2) \Id_n + (n-2)J_n + M$, and we
    may conclude that all eigenvalues of $\cR \cR^T$ are $(1\pm o(1))\cdot n^2$, which
    implies the same of $P_W = \cR^T \cR$.
    Since the range of $P_W$ and $\Pi_W$ is the same, we finally have that with probability $1-o(1)$
    \[
	(1+o(1))/n^2 \cdot P_W \succeq \Pi_W \succeq (1-o(1))/n^2 \cdot P_W,
    \]
    as desired.
\end{proof}

The following lemma allows us to approximate the projector to $V_0 \cup V_1$by a matrix that is easy to describe; we will use this matrix as an approximation to the projector in later proofs.
\begin{lemma}
    \label{lem:P01}
    Let $\Pi_{01}$ be the projection to the vector space $V_0 \cup V_1$.
    Let $P_{01}\in\R^{\binom{n}{2},\binom{n}{2}}$ be a matrix defined as follows:
    \[
	P_{01}(ab,cd) = \begin{cases}
	    \tfrac{2}{n-1} & |\{a,b,c,d\}| = 2\\
	    \tfrac{1}{n-1} & |\{a,b,c,d\}| = 3\\
	    0 & |\{a,b,c,d\}| = 4.
	\end{cases}
    \]
    Then
    \[
	\Pi_{01} \succeq P_{01} \succeq (\tfrac{n-2}{n-1})\cdot \Pi_{01},
    \]
\end{lemma}
\begin{proof}
    We will write down a basis for $V_0\cup V_1$, take a summation over its outer products, and then argue that this summation approximates $\Pi_{01}$.
    The vectors $v_1,\ldots,v_{n} \in \R^{\binom{n}{2}}$ are a basis for $V_1 \cup V_0$:
    \[
	v_i(a,b) = \begin{cases}
	    \frac{1}{\sqrt{n-1}} & \{a,b\} = \{i,\cdot\}\\
	    0 & \text{otherwise.}
	\end{cases}
    \]
    For any two $v_i,v_j$, we have $\iprod{v_i,v_j} = \tfrac{1}{n-1}$.
    Let $U \in \R^{n^2 \times n}$ be the matrix whose $i$th column is given by $v_i$.
    Notice that the eigenvalues of $\sum_{i} v_i v_i^T = UU^T$ are equal to the eigenvalues of $U^TU$, and that $U^TU = \tfrac{1}{n-1} J_n + \tfrac{n-2}{n-1}\Id_n$.
    Therefore,
    as both matrices have the same column and row spaces,
    \[
	\Pi_{01} \succeq \sum_{i} v_i v_i^T \succeq \tfrac{n-2}{n-1}\Pi_{01},
    \]
    Now, let $P_{01} = \sum_{i} v_i v_i^T$; we can explicitly calculate the entries of $P_{01}$,
    \[
	P_{01}(ab,cd) = \begin{cases}
	    \tfrac{2}{n-1} & |\{a,b,c,d\}| = 2\\
	    \tfrac{1}{n-1} & |\{a,b,c,d\}| = 3\\
	    0 & |\{a,b,c,d\}| = 4.
	\end{cases}
    \]
    The conclusion follows.
\end{proof}

We will require the fact that $W$ lies mostly outside of $V_0 \cup V_1$, which we prove in the following lemma.
\begin{lemma}
    \label{lem:01w}
    With probability at least $1 - O(n^{-\omega(\log n)})$,
    \[
	\|\Pi_{01}\Pi_W\Pi_{01}\| \le O\Paren{\frac{\log^2 n}{n}}.
    \]
\end{lemma}
\begin{proof}
    Call $M = \Pi_{01} \Pi_W \Pi_{01}$.
    We will apply the trace power method to $M$.
    By \pref{lem:P01} and \pref{lem:w-proj}, we may exchange $\Pi_W$ for $\frac{(1 + o(1))}{n^2} \sum_i a_i a_i^T$ and $\Pi_{01}$ for $P_{01}$.
    Letting $M'^k = (\tfrac{(1+o(1))}{n^2} P_{01} P_W )^k$, we have by the cyclic property of the trace that $\E[\Tr(M^k)] \le \Tr(M'^k)$.

    We consider the expression for $\E[\Tr(M'^k)]$.
    Let a \emph{chain} consists of a set of quadruples $\{a_\ell,b_\ell,c_\ell,d_\ell\}_{\ell\in[k]} \in [n]^4$ such that for each $\ell \in [k]$, we have $\|\{ a_\ell, b_\ell\} \cap \{c_{\ell-1},d_{\ell-1}\}| \ge 1$ (where we identify $a_\ell$ with $a_{\ell \mod k}$).
    Let $\cC_k$ denote the set of all chains of size $k$.
    We have that,
    $$
    \Tr(M^k) \le
    \Tr(M'^k)
    =
    \sum_{i_1,\ldots,i_k} \sum_{\{a_\ell,b_\ell,c_\ell,d_\ell\}_{\ell \in [k]} \in \cC_k} \prod_{\ell = 1}^k \frac{1 + o(1)}{n^2} \cdot r_\ell\cdot  A_{i_\ell, a_\ell} A_{i_\ell, b_\ell} A_{i_\ell, c_{\ell}} A_{i_\ell,d_{\ell}},
    $$
    where $r_\ell = \tfrac{1}{n-1}$ or $\tfrac{2}{n-1}$ depending on whether one or both of $a_\ell, b_\ell$ are common with the following link in the chain.
    The quantity $\Tr(M^k)$ consists of cycles of $k$ links, each link is a star on $4$ outer vertices $a_\ell, b_\ell, c_\ell, d_\ell$ with center vertex $i_\ell$, and the non-central vertices of the link must have at least one vertex in common with the next link, so each link has $4$ edges and the cycle is a connected graph.
    See the figure below for an illustration (dashed lines indicate vertex equality, and are not edges).
    
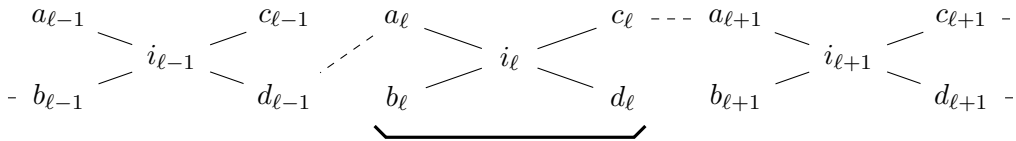
\begin{center}
    \begin{tikzpicture}[scale = 0.75]
    \def\hspc{2}
    \def\vspc{0.7}
    \def\shift{6}

\draw (\hspc, -\vspc) node[circle,fill=white] (d) {$d_\ell$}
--(0,0) node[circle,fill=white] (i) {$i_\ell$}
-- (\hspc, \vspc) node[circle,fill=white] (c) {$c_\ell$};
\draw (-\hspc, -\vspc) node[circle,fill=white] (b) {$b_\ell$}
-- (i)
-- (-\hspc, \vspc) node[circle,fill=white] (a) {$a_\ell$};
\draw (\shift+ \hspc, -\vspc) node[circle,fill=white] (d2) {$d_{\ell+1}$}
--(\shift,0) node[circle,fill=white] (i2) {$i_{\ell+1}$}
-- (\shift + \hspc, \vspc) node[circle,fill=white] (c2) {$c_{\ell+1}$};
\draw (\shift-\hspc, -\vspc) node[circle,fill=white] (b2) {$b_{\ell+1}$}
-- (i2)
-- (\shift -\hspc, \vspc) node[circle,fill=white] (a2) {$a_{\ell+1}$};
\draw[dashed] (c) -- (a2);
\draw[dashed] (c2) -- (\shift + \hspc + 1,\vspc);
\draw[dashed] (d2) -- (\shift + \hspc + 1,-\vspc);
\draw (-\shift+ \hspc, -\vspc) node[circle,fill=white] (d0) {$d_{\ell-1}$}
--(-\shift,0) node[circle,fill=white] (i0) {$i_{\ell-1}$}
-- (-\shift + \hspc, \vspc) node[circle,fill=white] (c0) {$c_{\ell-1}$};
\draw (-\shift-\hspc, -\vspc) node[circle,fill=white] (b0) {$b_{\ell-1}$}
-- (i0)
-- (-\shift -\hspc, \vspc) node[circle,fill=white] (a0) {$a_{\ell-1}$};
\draw[dashed] (a) -- (d0);
\draw[dashed] (b0) -- (-\shift - \hspc - 1,-\vspc);
\draw[very thick] (-1.2*\hspc,- \vspc - 0.5)--(-1.1*\hspc,- \vspc - 0.7) -- (1.1*\hspc,-\vspc-0.7) -- (1.2*\hspc,- \vspc - 0.5);
\end{tikzpicture}

\end{center}

    Each term in the product has a factor of at most $\tfrac{2(1+o(1))}{n^3}$, due to the scaling of the entries of $P_{01}$ and $P_W$.
    Thus we have
    $$
    \E[\Tr(M'^k)] \le \left(\frac{3}{n^3}\right)^k \sum_{i_1,\ldots, i_k} \sum_{\{a_\ell, b_\ell, c_\ell, d_\ell\}_{\ell\in[k]} \in \cC_n} \E\left[\prod_{\ell = 1}^k A_{i_\ell, a_\ell} A_{i_\ell, b_\ell} A_{i_\ell, c_\ell} A_{i_\ell, d_\ell}\right].
    $$
    The only contributing terms correspond to those for which every edge variable in the product has even multiplicity.
    Each contributing term is a connected graph and has $4k$ edges and at most $5k$ vertices where every labeled edge appears twice, so we may apply \pref{prop:eulerian} to conclude that there are at most $2k+1$ labels in any such cycle.
    We thus have that
    \[
	\E[\Tr(M'^k)] \le \Paren{\frac{3}{n^3}}^{k} \cdot n^{2k + 1} \cdot (5k)!,
    \]
    and applying \pref{lem:trace-power}, we conclude that $\|M\| \lesssim \frac{\log^2 n}{n}$ with probability $1 - O(n^{-\omega(\log n)})$, as desired.
\end{proof}

We combine the above lemmas to bound the norm of one final matrix that arises in the computations in \pref{thm:psd}.
\begin{lemma}\label{lem:W2}
    With probability $1-O(n^{-5})$,
    \[
	\sum_w D_w \Pi_2 \Pi_W \Pi_2 D_w \preceq O(n)\cdot \Pi_0 + O(\log^2 n)\cdot \Id_n.
    \]
\end{lemma}
\begin{proof}
    We begin by replacing $\Pi_2$ with $(1-\Pi_{01})$, as by \pref{lem:01w}, $\Pi_{01}$ can be replaced by $P_{01}$ which has a convenient form.
    For any vector $x \in \R^{\binom{n}{2}}$,
    \begin{align*}
	x^T\left(\sum_i D_i \Pi_2 \Pi_W \Pi_2 D_i\right)x
	&= x^T\Paren{\sum_i D_i \Pi_W  D_i}x - 2x^T\Paren{\sum_i D_i \Pi_{01}\Pi_W  D_i}x + x^T\Paren{\sum_i D_i \Pi_{01} \Pi_W \Pi_{01}D_i}x\\
	&\le \sum_{i} \left(\|\Pi_W D_i x\|^2 + 2\|\Pi_W\Pi_{01} D_i x\|\cdot \|\Pi_W D_i x\| + \|\Pi_W\Pi_{01}D_i x\|^2\right)\\
	&\le 2 x^T\Paren{\sum_{i} D_i \Pi_W D_i} x + 2\Paren{\sum_i (D_ix)^T \Pi_{01} \Pi_W \Pi_{01} D_ix}\\
	&\le 2 x^T\Paren{\sum_{i} D_i \Pi_W D_i}x +2 n \left\| \Pi_{01} \Pi_W\Pi_{01}\right\| \cdot \|x\|^2,
    \end{align*}
    where to obtain the second line we have applied Cauchy-Schwarz, to obtain the third line we have used the fact that $a^2 + b^2 \ge 2ab$, and to obtain the final line we have used the fact that $\|D_i x\| = \|x\|$.

    Now,  the second term is $O(\log^2 n)\cdot \|x\|^2$ with overwhelming probability by \pref{lem:01w}.
    It remains to bound the first term.
    To this end, we apply \pref{lem:w-proj} to replace $\Pi_W$ with $\tfrac{1+o(1)}{n^2}\cdot P_W =\tfrac{1+o(1)}{n^2}\cdot \sum_{i} a_i a_i^T$.
Let $M = \frac{1}{n^2}\cdot\sum_{i} D_i P_W D_i$. An entry of $M$ has the form
    \[
	M(ab,cd) = \frac{1}{n^2}\Paren{ n + \sum_{i \neq j} A_{ia}A_{ib} A_{ic} A_{id} A_{ja} A_{jb} A_{jc} A_{jb}}.
    \]
    Thus we can see that $M = \tfrac{1}{n} J_{\binom{n}{2}} + \tfrac{1+o(1)}{n^2} BB^\top$, where $J_{\binom{n}{2}}$ is the all-ones matrix in $\R^{\binom{n}{2}\times \binom{n}{2}}$ and $B$ is the matrix whose entries have the form
    \[
	B(ab,ij) =  A_{ia} A_{ib} A_{ja} A_{jb}.
    \]
    The matrix $B$ is actually equal to the matrix $J_{4,1}$ from \cite{DeshpandeM15}, and by \pref{lem:wigner} has $\|B\|\lesssim \barn $ with probability $1-O(n^{-5})$.
    We can thus conclude that with probability $1-O(n^{-5})$,
    $\|M - \tfrac{1}{n} J_{\binom{n}{2}}\| \le \tfrac{1 + o(1)}{n^2}\|B\|^2 \le \tilde O(1)$,
    and so $x^T M x \le \tfrac{1+o(1)}{n} \iprod{x,\ovec_{\binom{n}{2}}}^2 + x^T(M-n^{-1}J)x \le O(n) \cdot \|\Pi_0 x\|^2 + \tilde{O}(1)\cdot\|x\|^2$, which gives the desired result.
\end{proof}

\section*{Acknowledgements}
We thank Satish Rao for many helpful conversations. \hfill

We also greatfully acknowledge the comments of anonymous reviewers in helping us improve the manuscript.
\addreferencesection
\bibliographystyle{amsalpha}
\bibliography{writeup}

\appendix
\section{Matrix Norm Bounds from Deshpande and Montanari}\label{app:dm-bounds}
In this appendix, we give for completeness a list of the bounds proven by Deshpande and Montanari \cite{DeshpandeM15} that were not included in the body above out of space or expository considerations.

\begin{definition}\label{def:Js}

    Let $A = \{a,b\} \subset [n] $ be disjoint from $B = \{c,d\}\subset [n]$.
    For $\eta \in \{1,\ldots,4\}$ and for $\nu(\eta) \in [\binom{4}{\eta}]$ we define the matrices $\tilde J'_{\eta,\nu(\eta)}$ as follows:
	\[
    \begin{array}{llll}
{\tilde J'}_{1,1}(A,B) = A_{ac}
&
{\tilde J'}_{2,1}(A,B) = A_{ac}A_{bd}
&
{\tilde J'}_{3,1}(A,B) = A_{ac}A_{ad}A_{bc}
&
{\tilde J'}_{4,1}(A,B) = A_{ac} A_{ad} A_{bc} A_{bd}
\\
{\tilde J'}_{1,2}(A,B) = A_{ad}
    &
{\tilde J'}_{2,2}(A,B) = A_{ac}A_{bc}
    &
{\tilde J'}_{3,2}(A,B) = A_{ac}A_{bc}A_{bd}
&\\
{\tilde J'}_{1,3}(A,B) = A_{bc}
&
{\tilde J'}_{2,3}(A,B) = A_{ac}A_{ad}
&
{\tilde J'}_{3,3}(A,B) = A_{ac}A_{ad}A_{bd}
&\\
{\tilde J'}_{1,4}(A,B) = A_{bd}
&
{\tilde J'}_{2,4}(A,B) = A_{ad}A_{bd}
&
{\tilde J'}_{3,4}(A,B) = A_{ad}A_{bc}A_{bd}
&\\
&
{\tilde J'}_{2,5}(A,B) = A_{bc}A_{bd}
&&\\
&
{\tilde J'}_{2,6}(A,B) = A_{ad}A_{bc}
&&
\end{array}
\]
And further, letting $\cP: \R^{\binom{n}{2} \times \binom{n}{2}}$ to be the matrix projector such that
\[
    (\cP M)_{A,B} = \begin{cases}
	M_{A,B} & |A \cup B| = 4\\
	0 &\text{otherwise}\mcom
    \end{cases}
\]
we define $J'_{\eta,nu} = \cP J_{\eta,\nu}$,
and finally we define $J_{\eta,\nu} \defeq \tfrac{2^\eta}{16}\alpha_4 \cdot J_{\eta,\nu}'$ (as in Deshpande and Montanari), so that $Q = \sum_{\eta=1}^4\sum_{\nu=1}^{\binom{4}{\eta}} J_{\eta,\nu}$.

\end{definition}

Notice that since we have defined $A_{ii} = 0$ and since $|\{a,b\}| = |\{c,d\}| = 2$, we have
$\tilde J_{4,1} = J_{4,1}$.
For some of the terms, the $\tilde J$ is never considered; however for some terms it is cleaner to bound the spectral norm of $\tilde J$ in the subspace $V_2$, and so Deshpande and Montanari provide trace power method bounds on the difference in norm:
\begin{lemma}[Lemma 4.26 in \cite{DeshpandeM15-arxiv}]\label{lem:tildes}
    With probability at least $1 - 6n^{-5}$, for each $\eta \le 2$ and for each $\nu \le \binom{4}{\eta}$,
    \[
	\|J_{\eta,\nu} - \tilde J_{\eta,\nu}\| \lesssim \alpha_4 \bar{n}.
    \]
\end{lemma}

\bigskip
Deshpande and Montanari use the trace power method to bound the norm of $Q$ by bounding the norms of the $J_{\eta,\nu}$ individually.
Some of the $J_{\eta,\nu}$ matrices have Wigner-like behavior.
\begin{lemma}[Lemmas 4.21, 4.22 in \cite{DeshpandeM15-arxiv}]\label{lem:wigner}
    With probability $1 - O(n^{-5})$, we have that for each $(\eta,\nu) \in \{(2,1),(2,6),(3,\cdot),(4,1)\}$,
    \[
	\left\| J_{\eta,\nu} \right\| \lesssim \alpha_4 \cdot \barn.
    \]
\end{lemma}

A select few of the $J_{\eta,\nu}$ have larger eigenvalues.
\begin{lemma}[Lemmas 4.23, 4.24 in \cite{DeshpandeM15-arxiv}]\label{lem:not-wigner}
    With probability $1 - O(n^{-4})$, we have that for each $(\eta,\nu) \in \{(1,\cdot),(2,2),(2,3),(2,4),(2,5)\}$,
    \[
	\left\| J_{\eta,\nu} \right\| \lesssim \alpha_4 \cdot \barn^{3/2}.
    \]
\end{lemma}

We also give a short proof of an observation of Deshpande and Montanari, which states that some of the $J_{\eta,\nu}$ vanish when projected to $V_2$:
\begin{observation}[Lemmas 4.23, 4.24 in \cite{DeshpandeM15}]\label{obs:nullspace}
    Let $\Pi_2$ be the projector to $V_2$.
    Then always,
    \[
	\left\|\Pi_2 \Paren{\sum_{\nu = 1}^4 \tilde J_{1,\nu}}\right\| = 0,\quad \text{and similarly,}\qquad
	\|\Pi_2 (\tilde J_{2,3} + \tilde J_{2,5})\| = 0.
    \]
\end{observation}
\begin{proof}
    The proof follows from noting that the range of both of these sums of $J_{\eta,\nu}$ is in $V_1$.
    Consider some vector $v \in \R^{\binom{n}{2}}$;
    let $v' \defeq \left(\sum_{\nu = 1}^4 J_{1,\nu}\right) v$.
    We will look at the entry of $v'$ indexed by the disjoint pair $A = \{a,b\}$.
    By definition of the $J_{1,\nu}$, we have that
    \begin{align*}
	v'_{A} &= \sum_{c,d\in[n]} \bigg((A_{a,c} + A_{a,d}) + ( A_{b,c} + A_{b,d}\bigg)v_{c,d}\\
				       &= \Paren{\sum_{c,d} ( A_{a,c} +  A_{a,d})v_{c,d}}
				       + \Paren{\sum_{c,d} ( A_{b,c} +  A_{b,d})v_{c,d}},
    \end{align*}
    and so by the characterization of $V_1$ from \pref{prop:eigenspaces} the vector $v' \in V_1$.
    The conclusion follows.

    A similar proof holds for the matrix $\tilde J_{2,3} + \tilde J_{2,5}$.
\end{proof}

Finally, we use a bound on the norm of the matrix $K$, which is the difference of $H_{2,2}$ and the non-multilinear entries.
\begin{lemma}[Lemma 4.25 in \cite{DeshpandeM15-arxiv}] \label{lem:K-term}
    Let $K$ be the restriction of $H_{2,2} - \E[H_{2,2}]$ to entries indexed by sets of size at most $3$.
    With probability at least $1 - n^{-5})$,
    \[
	\|K\| \le \tilde{O}(\alpha_3 n^{1/2}).
    \]
\end{lemma}

We also require bounds on the matrices used in the Schur complement steps.
The bounds of Deshpande and Montanari suffice for us, since we do not modify moments of order less than $4$.
\begin{lemma}[Consequence of Proposition 4.19 in \cite{DeshpandeM15-arxiv}]\label{lem:H11}
    Define $Q_n \in \R^{n,n}$ to be the orthogonal projection to the space spanned by $\vec{1}$.
    Suppose that $\underline{\alpha}$ satisfies $\alpha_1 - \alpha_2 \ge \Omega(\alpha_2 n^{-1/2})$ and $\alpha_2 - 2\alpha_1^2 \ge 0$, $\alpha_1 \ge 0$.
    Then with probability at least $1-n^{-5}$,
    \begin{align*}
	H_{1,1} &\succeq 0\\
	H_{1,1}^{-1} &\preceq \frac{1}{n(\alpha_2 p - \alpha_1^2 p)} Q_n + \frac{2}{\alpha_1} Q_n^\perp.
    \end{align*}
\end{lemma}

\subsection{Additional Proofs}\label{app:misc}
We prove \pref{lem:probsubspace}, which follows almost immediately from the bounds of \cite{DeshpandeM15}.

\begin{proof}[Proof of \pref{lem:probsubspace}]
    Using the matrices from \pref{def:Js} and \pref{obs:nullspace}, we have that
    \begin{align*}
	\Pi_2 Q &= \Pi_2(\tilde J_{2,4} + \tilde J_{2,2}) + \Pi_2 \Paren{J_{2,4} - \tilde J_{2,4} + J_{2,2} - \tilde J_{2,2}} + \Pi_2\Paren{\sum_{\nu = 1}^4J_{3,\nu} + J_{4,1}},\\
	\Pi_2 \Pi_W Q &= \Pi_2\Pi_W(\tilde J_{2,4} + \tilde J_{2,2}) + \Pi_2 \Pi_W\Paren{J_{2,4} - \tilde J_{2,4} + J_{2,2} - \tilde J_{2,2}} + \Pi_2\Pi_W\Paren{\sum_{\nu = 1}^4J_{3,\nu} + J_{4,1}},
    \end{align*}
    where we have used the fact that the columns of $J_{2,4}$ and $J_{2,2}$ lie in $W$.
    We thus have
    \begin{align*}
	\Pi_2 Q - \Pi_2\Pi_WQ
	&= \Pi_2(I - \Pi_W)\Paren{J_{2,2} - \tilde J_{2,2} + J_{2,4} - \tilde J_{2,4} + \sum_{\nu = 1}^4 J_{3,\nu} + J_{4,1}},
    \end{align*}
    And by the bounds $\|J_{3,\cdot}\| \le \alpha_4 \cdot \barn $ and $\|J_{4,1}\| \le \alpha_4 \cdot \barn$ from \pref{lem:wigner} and the bounds $\|J_{2,4} - \tilde J_{2,4}\| \le \alpha_4 \barn$ and $\|J_{2,2} - \tilde J_{2,2}\| \le \alpha_4 \barn$ from \pref{lem:tildes}, and because $\Pi_2(I - \Pi_W)$ is a projection, the conclusion follows.
\end{proof}

Now, we prove that the trace power method works, for completeness.
\begin{proof}[Proof of \pref{lem:trace-power}]
    The proof follows from an application of Markov's inequality.
    We have that for even $k$,
    \begin{align*}
	\Pr[\|M\| \ge t]
	&= \Pr[\|M^k\| \ge t^k]\\
	&\le \Pr[\Tr(M^k) \ge t^k]\\
	&\le \frac{1}{t^k}\E[\Tr(M^k)]\\
	&\le \frac{1}{t^k} \sqrt{\pi \gamma k} \Paren{\frac{\gamma k}{e}}^k n^{\alpha k + \beta}\mcom
    \end{align*}
    where we have applied Stirling's approximation in the last step.
    Choosing $k = O(\log n)$ and $t = O\Paren{\eta^{-1/k}\cdot \gamma \cdot \log n \cdot n^{\alpha}}$ completes the proof.
\end{proof}

\fi

\ifnum\full=0
\section*{Acknowledgements}
We thank Satish Rao for many helpful conversations.  \hfill

We also greatfully acknowledge the comments of anonymous reviewers in helping us improve the manuscript.

\addreferencesection
\bibliographystyle{amsalpha}
\bibliography{pc}

\appendix

\subsection{Additional Proofs}\label{app:misc}
We prove \pref{lem:probsubspace}, which follows almost immediately from the bounds of \cite{DeshpandeM15}.

\begin{proof}[Proof of \pref{lem:probsubspace}]
    Using the matrices from \pref{def:Js} and \pref{obs:nullspace}, we have that
    \begin{align*}
	\Pi_2 Q &= \Pi_2(\tilde J_{2,4} + \tilde J_{2,2}) + \Pi_2 \Paren{J_{2,4} - \tilde J_{2,4} + J_{2,2} - \tilde J_{2,2}} + \Pi_2\Paren{\sum_{\nu = 1}^4J_{3,\nu} + J_{4,1}},\\
	\Pi_2 \Pi_W Q &= \Pi_2\Pi_W(\tilde J_{2,4} + \tilde J_{2,2}) + \Pi_2 \Pi_W\Paren{J_{2,4} - \tilde J_{2,4} + J_{2,2} - \tilde J_{2,2}} + \Pi_2\Pi_W\Paren{\sum_{\nu = 1}^4J_{3,\nu} + J_{4,1}},
    \end{align*}
    where we have used the fact that the columns of $J_{2,4}$ and $J_{2,2}$ lie in $W$.
    We thus have
    \begin{align*}
	\Pi_2 Q - \Pi_2\Pi_WQ
	&= \Pi_2(I - \Pi_W)\Paren{J_{2,2} - \tilde J_{2,2} + J_{2,4} - \tilde J_{2,4} + \sum_{\nu = 1}^4 J_{3,\nu} + J_{4,1}},
    \end{align*}
    And by the bounds $\|J_{3,\cdot}\| \le \alpha_4 \cdot \barn $ and $\|J_{4,1}\| \le \alpha_4 \cdot \barn$ from \pref{lem:wigner} and the bounds $\|J_{2,4} - \tilde J_{2,4}\| \le \alpha_4 \barn$ and $\|J_{2,2} - \tilde J_{2,2}\| \le \alpha_4 \barn$ from \pref{lem:tildes}, and because $\Pi_2(I - \Pi_W)$ is a projection, the conclusion follows.
\end{proof}

Now, we prove that the trace power method works, for completeness.
\begin{proof}[Proof of \pref{lem:trace-power}]
    The proof follows from an application of Markov's inequality.
    We have that for even $k$,
    \begin{align*}
	\Pr[\|M\| \ge t]
	&= \Pr[\|M^k\| \ge t^k]\\
	&\le \Pr[\Tr(M^k) \ge t^k]\\
	&\le \frac{1}{t^k}\E[\Tr(M^k)]\\
	&\le \frac{1}{t^k} \sqrt{\pi \gamma k} \Paren{\frac{\gamma k}{e}}^k n^{\alpha k + \beta}\mcom
    \end{align*}
    where we have applied Stirling's approximation in the last step.
    Choosing $k = O(\log n)$ and $t = O\Paren{\eta^{-1/k}\cdot \gamma \cdot \log n \cdot n^{\alpha}}$ completes the proof.
\end{proof}

\fi

\end{document}